\documentclass[notitlepage,a4paper,superscriptaddress,10pt,tightenlines,twocolumn,pra,nofootinbib]{revtex4-2}
\usepackage[utf8]{inputenc}
\usepackage{graphicx}
\usepackage{physics, amsmath, amssymb, amsthm, units, dsfont, bm}
\usepackage{mathtools, amsfonts, mathrsfs, bbm}

\usepackage[dvipsnames]{xcolor}
\usepackage{adjustbox}
\usepackage{svg}

\usepackage{booktabs}
\usepackage{tabularx}

\newcolumntype{C}[1]{>{\centering\arraybackslash}p{#1}}

\usepackage{tikz}
\usepackage{tikzsymbols}
\usetikzlibrary{math,calc}
\usetikzlibrary{overlay-beamer-styles}
\usetikzlibrary{arrows.meta}
\usetikzlibrary{automata, positioning, arrows}
\usetikzlibrary{backgrounds}
\usetikzlibrary{decorations.markings}

\colorlet{alice}{Green}
\colorlet{bob}{Blue}
\colorlet{charlie}{Red}

\definecolor{dark-gray}{rgb}{.35,.55,.55}
\definecolor{dark-blue}{rgb}{.0,.0,.6}
\usepackage[colorlinks=true,linkcolor=dark-blue,citecolor=dark-blue,urlcolor=dark-blue]{hyperref}
\usepackage{url}

\newcommand{\kommentar}[1]{}

\renewcommand{\ket}[1]{| #1 \rangle}
\renewcommand{\bra}[1]{\langle #1 |}

\renewcommand{\braket}[2]{\langle #1 | #2 \rangle}

\renewcommand{\varrho}{\rho}

\theoremstyle{plain}
\newtheorem{theorem}{Theorem}
\numberwithin{theorem}{section}
\newtheorem{lemma}[theorem]{Lemma}

\newtheorem{proposition}[theorem]{Proposition}
\newtheorem{corollary}[theorem]{Corollary}
\newtheorem{conjecture}[theorem]{Conjecture}

\newtheorem{definition}[theorem]{Definition}
\newtheorem{property}[theorem]{Property}

\theoremstyle{remark}

\makeatletter
\def\maketitle{
\@author@finish
\title@column\titleblock@produce
\suppressfloats[t]}
\makeatother

\begin{document}

\title{Asymptotic entanglement in circle stabilizer states\\ and states forbidding arbitrary vertex-minors}

\author{Kenneth Goodenough
}
\affiliation{Naturwissenschaftlich-Technische Fakult\"{a}t, Universit\"{a}t Siegen, Walter-Flex-Stra\ss e 3, 57068 Siegen, Germany}

\author{
Marcelo Sales
}
\affiliation{Department of Mathematics, California Institute of Technology, Pasadena, CA 91125, USA}

\date{\today}

\begin{abstract}
Stabilizer states play a central role in quantum information theory, and understanding their entanglement has motivated a large body of work. A well-studied question in particular is when a stabilizer state $\ket{\psi}$ can be transformed into another stabilizer state $\ket{\phi}$ using only single-qubit Clifford operations and Pauli measurements. If this is possible, we say that $\ket{\phi}$ is a \emph{vertex-minor} of $\ket{\psi}$. Assuming Geelen's weak structural conjecture on vertex-minors, we establish the following general statement. For any fixed stabilizer state $\ket{\phi}$, the entanglement in stabilizer states $\ket{\psi}$ that do not contain $\ket{\phi}$ as a vertex-minor is asymptotically constrained. More concretely, we show that the distance of any sufficiently \emph{rank-connected} $\ket{\psi}$ not containing $\ket{\phi}$ as a vertex-minor grows as $O(\log n)$, and prove similar results for the so-called locally accessible information, a quantity that captures the amount of information that can be learned through single-qubit Pauli measurements. Our results rely on (i) connecting the above two entanglement measures to rank functions of \emph{multimatroids}, (ii) connecting the rank functions of circle stabilizer states to rank functions on $4$-regular multigraphs, which asymptotically constrains the entanglement of circle stabilizer states, and (iii) using Geelen's weak structural conjecture on vertex-minors to `lift' the previous result to sufficiently connected states in proper vertex-minor-closed families of stabilizer states. Our results establish a connection between asymptotic stabilizer entanglement and forbidden vertex-minors, with direct implications for the entanglement that can be generated in quantum devices.
\end{abstract}




\maketitle

\section{Introduction}
Entanglement is one of the cornerstones of quantum information theory, where in particular the understanding of entanglement in stabilizer states has led to both theoretical and practical insights~\cite{gottesman1997stabilizer, grassl2002graphs, markham2008graph, shettell2020graph, hein2004multiparty, goodenough2024near}. While prior work has shown that the entanglement in stabilizer states can be large, such works typically assume very few constraints---if any---on the states under consideration.

In practice, such constraints can arise from limited long-range interactions~\cite{brandhofer2025hardware} or a limited number of physical devices such as photonic emitters~\cite{li2022photonic}. Any experimental setup must contend with such constraints while trying to generate highly entangled states. From a theoretical perspective, it is important to understand under which conditions strongly entangled states can---and cannot---exist.

We show, assuming a widely believed conjecture in combinatorics, that there is a strong dichotomy; any  constraint on the so-called \emph{vertex-minors} strongly limits the asymptotic entanglement in stabilizer states. 

To make the above statement precise, we will need to quantify what we mean by entanglement and constraints/vertex-minors. To quantify entanglement, we use two entanglement measures: the distance $d$ and the locally accessible information $\alpha_{\rm loc}$. The distance (closely connected to the $k$-uniformity) is a well-known quantity, corresponding to the smallest number of qubits whose marginal is not maximally mixed~\cite{javelle2012minimum, cattaneo2015mindeg, raissi2022general}.

\begin{definition}
The distance $d$ of a stabilizer state $\ket{\psi}$ is the weight of the smallest non-identity stabilizer in the stabilizer group of $\ket{\psi}$.
\end{definition}

To motivate the other measure, we use the following framing. A stabilizer state on $n$ qubits has an associated stabilizer basis; these are all $2^n$ stabilizer states that share the same stabilizers, but differ in the phases in front of the stabilizers. A choice of one of the $2^n$ states encodes $n$ bits of information; the locally accessible information $\alpha_{\rm loc}$ is concerned with how many bits of information can be extracted under single-qubit Pauli measurements (SQPMs).

\begin{definition}[informal]
The locally accessible information $\alpha_{\rm loc}$ of a stabilizer state $\ket{\psi}$ is defined as the maximum number of bits that can be learned when measuring $\ket{\psi}$ with SQPMs and classical communication.
\end{definition}

One hallmark of entanglement is that local measurements cannot extract all information present in the global state. As such, strongly entangled states have small $\alpha_{\rm loc}$. With this in mind, we will still refer to $\alpha_{\rm loc}$ as an entanglement measure. The locally accessible information has been studied in~\cite{markham2007entanglement}, where it has been used to bound the geometric measure of entanglement.

We summarize below (bounds on) the optimal asymptotic behavior of the two entanglement measures.

\begin{proposition}\label{prop:meas_for_arbitrary_stab}
Optimized over all stabilizer states on $n$ qubits, the quantities $d$ and $\alpha_{\rm loc}$ satisfy
\begin{align}
\max_{\psi}d(\ket{\psi})=&~\Theta\left( n\right), \nonumber \\
\min_{\psi}\alpha_{\rm loc}(\ket{\psi})=&~O\left(\sqrt{n}\right). \nonumber \\
\end{align}
This follows from~\cite{javelle2012minimum} and~\cite{ascoli2026almost}, respectively.
\end{proposition}

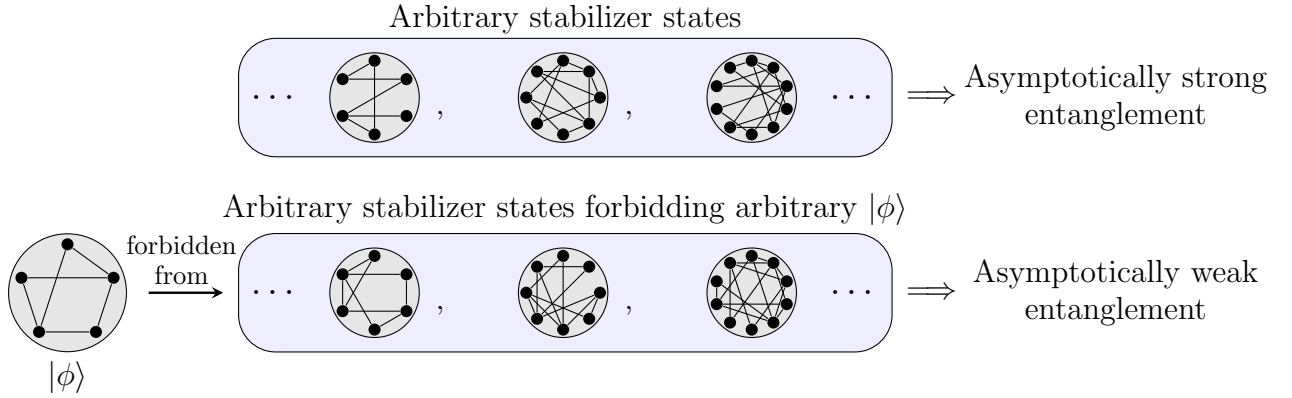
\begin{figure*}[t]
    \centering
    \begin{tikzpicture}[
    scale=0.58,
    >=stealth,
    vertex/.style={circle, fill=black, inner sep=1.6pt},
    seqborder/.style={fill=gray!10!blue!7, draw=black, rounded corners=14pt},
    graphdisk/.style={fill=gray!20, draw=black}
]

\def\BaseSeed{213}

\def\R{1.02}        
\def\r{0.84}        
\def\TargetR{1.12}
\def\Targetr{0.92}

\newcommand{\RandomBorderGraph}[6]{%
\pgfmathtruncatemacro{\N}{#3}
\pgfmathtruncatemacro{\NmOne}{#3-1}
\pgfmathtruncatemacro{\SeedValue}{#4}
\begin{scope}[shift={(#1,#2)}]
    \filldraw[graphdisk] (0,0) circle (\R cm);

    \foreach \i in {1,...,\N}{
        \pgfmathsetmacro{\ang}{90 + 360/\N*(\i-1)}
        \node[vertex] (#6\i) at (\ang:\r) {};
    }

    \pgfmathsetseed{\SeedValue}
    \foreach \i in {1,...,\NmOne}{
        \pgfmathtruncatemacro{\jstart}{\i+1}
        \foreach \j in {\jstart,...,\N}{
            \pgfmathparse{rnd}
            \ifdim\pgfmathresult pt<#5pt
                \draw (#6\i)--(#6\j);
            \fi
        }
    }
\end{scope}
}

\draw[seqborder] (-3.1,1.2) rectangle (11.8,3.9);
\draw[seqborder] (-3.1,-3.25) rectangle (11.8,-0.55);

\node[font=\large] at (4.35,4.35) {Arbitrary stabilizer states};
\node[font=\large] at (4.35,0.00) {Arbitrary stabilizer states forbidding arbitrary $\ket{\phi}$};

\node[font=\large, align=center] at (12.7,2.55) {$\implies$};
\node[font=\large, align=center] at (16.9,2.55) {Asymptotically strong\\entanglement};

\node[font=\large, align=center] at (12.7,-1.9) {$\implies$};
\node[font=\large, align=center] at (16.9,-1.9) {Asymptotically weak\\entanglement};

\node[scale=1.5] at (-2.25,2.55) {$\cdots$};
\node[scale=1.2] at (1.5,2.15) {$,$};
\RandomBorderGraph{0}{2.55}{6}{\BaseSeed+13}{0.38}{A}
\node[scale=1.2] at (5.8,2.15) {$,$};
\RandomBorderGraph{4.3}{2.55}{8}{\BaseSeed+28}{0.4}{B}
\RandomBorderGraph{8.6}{2.55}{10}{\BaseSeed+40}{0.26}{C}
\node[scale=1.5] at (10.95,2.55) {$\cdots$};

\node[scale=1.5] at (-2.25,-1.9) {$\cdots$};
\RandomBorderGraph{0}{-1.9}{6}{\BaseSeed+46}{0.34}{D}
\node[scale=1.2] at (1.5,-2.3) {$,$};
\RandomBorderGraph{4.3}{-1.9}{8}{\BaseSeed+62}{0.32}{E}
\RandomBorderGraph{8.6}{-1.9}{10}{\BaseSeed+68}{0.32}{F}
\node[scale=1.2] at (5.8,-2.3) {$,$};
\node[scale=1.5] at (10.95,-1.9) {$\cdots$};

\draw[->, line width=0.9pt]
     (-4.25-0.9,-1.9) -- (-2.65-0.8,-1.9);
\node[font=\normalsize, align=center] at (-3.3-1.1,-1.15) {forbidden\\from};

\begin{scope}[shift={(-6.3-0.7,-1.9)}]
    \filldraw[graphdisk] (0,0) circle (\TargetR*1.2 cm);

    \foreach \i/\ang in {
        1/90,
        2/162,
        3/234,
        4/306,
        5/18
    }{
        \node[vertex] (T\i) at (\ang:\Targetr*1.2) {};
    }

    \draw (T1)--(T3);
    \draw (T1)--(T5);
    \draw (T2)--(T3);
    \draw (T2)--(T5);
    \draw (T3)--(T4);
    \draw (T4)--(T5);

    \node[font=\large] at (0,-1.92) {$\ket{\phi}$};
\end{scope}

\end{tikzpicture}
    \vspace*{-3mm}
\caption{Graphical description of theorem \ref{thm:main_result_informal}. Top depicts a sequence of arbitrary stabilizer states (which we depict as graph states for convenience), bottom depicts a sequence of stabilizer states, but that do not contain some arbitrary but fixed $\ket{\phi}$ as a vertex-minor. The asymptotic entanglement, as quantified by the distance $d$ and locally accessible information $\alpha_{\rm loc}$ can grow strongly (see main text for more details) for the top sequence, while the entanglement in any sequence of (sufficiently strongly rank-connected) stabilizer states forbidding $\ket{\phi}$ is necessarily constrained.
}
\label{fig:overview}
\end{figure*}

We now formally define what we mean by a constraint, which we phrase in terms of \emph{vertex-minors}.

\begin{definition}[informal]
A stabilizer state $\ket{\phi}$ is a vertex-minor of a stabilizer state $\ket{\psi}$ if $\ket{\psi}$ can be transformed into $\ket{\phi}$ by SQPMs and single-qubit Cliffords. The measured qubits are traced out, and the single-qubit Cliffords may depend on the measurement outcomes. 
\end{definition}
We will write $\ket{\psi} \succ \ket{\phi}$ to mean that $\ket{\psi}$ has a vertex-minor isomorphic to $\ket{\phi}$, where here isomorphism means up to qubit relabelling. Conversely, $\ket{\psi} \nsucc \ket{\phi}$ is short-hand for  $\ket{\psi}$ not having a vertex-minor isomorphic to $\ket{\phi}$. Informally, $\ket{\psi}\succ \ket{\phi}$ means that $\ket{\phi}$ is a `substructure' of $\ket{\psi}$, while $\ket{\psi}\nsucc \ket{\phi}$ implies this is not the case.
Forbidding one or more states $\ket{\phi_i}$ as a vertex-minor is an abstract way of modeling a constraint: no matter how an experimental setup is scaled up, it should not be able to create any of the states $\ket{\phi_i}$.

We now state our main result informally.

\begin{theorem}[informal]\label{thm:main_result_informal}
Fix an arbitrary stabilizer state $\ket{\phi}$. 
Let $\ket{\psi}$ be a \emph{sufficiently rank-connected} $n$-qubit stabilizer state such that $\ket{\psi}\nsucc \ket{\phi}$. Conditional on \emph{Geelen's weak structural conjecture}, it holds that
\begin{align}
d(\ket{\psi})=&~O\left(\log n\right), \nonumber \\
\alpha_{\rm loc}(\ket{\psi})=&~\Omega\left(\frac{n}{\log n}\right) \ . \nonumber \\
\end{align}

\end{theorem}

We provide further details on Geelen's weak structural conjecture, and what \emph{sufficiently rank-connected} means in Section \ref{sec:main_geelen}. We suspect that it is possible to get rid of the sufficiently $k$-rank-connected quantifier.

Informally, theorem \ref{thm:main_result_informal} says that any sequence of devices that cannot prepare a given fixed state $\ket{\phi}$, also cannot prepare strongly entangled states (that are sufficiently rank-connected). Furthermore, sufficiently rank-connected stabilizer states $\ket{\psi}$ with strong entanglement necessarily need to contain all sufficiently small states $\ket{\phi}$ as a vertex-minor, suggesting a close relationship between entanglement of stabilizer states and their substructures.

\section{Techniques}\label{sec:techniques}
To prove our main result, we show a connection between circuit-partitions of $4$-regular multigraphs and measurement statistics of single-qubit Pauli measurements on a subset of stabilizer states called \emph{circle stabilizer states}. These measurement statistics are described by the \emph{rank} and \emph{nullity}.

\subsection{Rank and nullity}
Stabilizer states are particularly well-behaved under SQPMs. Namely, conditioned on the outcomes of other SQPMs, an SQPM on a stabilizer state yields either a deterministic or completely random outcome. As such, a Pauli string $P$ encoding a choice of SQPMs on a subset of qubits of a stabilizer state yields a probability distribution over outcomes, such that the probability distribution has an entropy equal to an integer. This integer we call the \emph{rank} $r(P)$ of $P$ (with respect to a fixed $\ket{\psi}$). It will also be convenient to define the nullity of $P$ as $\nu(P)\equiv w(P)-r(P)$, where $w(P)$ is the weight of $P$, i.e.~the number of non-identity elements in $P$.

Let us provide a brief example. For a GHZ state $\ket{000}+\ket{111}$, measuring $P = ZZZ$ yields either the outcome $000$ or $111$ uniformly at random. This distribution has an entropy of $1$, and thus $r(ZZZ)=1$ and $\nu(ZZZ)=2$. On the other hand, it can be checked that $r(XXX) = 2$ and $\nu(XXX)=1$, see Appendix \ref{sec:K3_example} for more details.

Just as expectation values of all Pauli strings uniquely determine a state, so does the rank function uniquely determine a stabilizer state up to Pauli unitaries, see Proposition~\ref{prop:unique_rank}. Furthermore, the two entanglement measures $d$ and $\alpha_{\rm loc}$ are conveniently expressed in terms of the rank/nullity function.

\begin{proposition}
The distance of a stabilizer state is the smallest weight $w(P)$ of a Pauli string for which $\nu(P) =1$.
\end{proposition}

A Pauli string with weight $n$ we call \emph{complete}.

\begin{proposition}
The locally accessible information $\alpha_{\rm loc}$ of a stabilizer state with nullity function $\nu$ is the largest nullity $\nu(P)$ of a complete Pauli string $P$.
\end{proposition}

We make the following observation, which we believe to be of independent interest.

\begin{property}
    Rank functions of stabilizer states satisfy the axioms of the rank function of a \emph{multimatroid}. 
\end{property}
Multimatroids have been studied in the mathematical community~\cite{bouchet1997multimatroids, bouchet1998multimatroids, bouchet2001multimatroids, bouchet19984multimatroids, merino2025activities, brijder2015isotropic1, noble2025primer, brijder2014interlace}, and generalize the notion of \emph{matroids}~\cite{welsh2010matroid}, see the Appendix for further details.

\subsection{Circle stabilizer states}
Circle graphs are a well-studied class of graphs in the mathematical literature~\cite{kloks1993treewidth, davies2019circle}, where they are in particular relevant from the perspective of vertex-minors~\cite{geelen2023grid, campbell2026erdHos}. Due to this, circle graph states have recently received a lot of attention in the quantum literature~\cite{dahlberg2018transforming, harrison2025fermionic, bhatti2025distributing, hahn2026structure}.

Circle graphs are often defined as the intersection graphs of a finite number of chords in a circle~\cite{davies2019circle}. Instead, we characterize circle graph states as those graph states whose rank functions can be described in terms of an associated \emph{$4$-regular multigraph $F$}.  This characterization connects directly with local measurement statistics, generalizes naturally to circle stabilizer states (i.e.~stabilizer states locally equivalent to circle graph states), and allows us to prove statements about $d$ and $\alpha_{\rm loc}$ for circle stabilizer states by exploiting the structure of $4$-regular multigraphs.

\begin{definition}
A $4$-regular multigraph $F$ is a graph where parallel edges and self loops are allowed, and where each vertex has degree $4$.
\end{definition}

Note that for each vertex $v\in V$ of a $4$-regular multigraph $F$, there are three possible pairings of the four incident half-edges to $v$. We label the three possible choices arbitrarily by $X_v, Y_v, Z_v$. 

\begin{definition} [informal]
Let $F$ be a $4$-regular multigraph on vertex set $V$. For a subset $W\subseteq V$, a choice of pairing $P_v \in \lbrace X_v, Y_v, Z_v\rbrace$ for each vertex $v\in W$ is called a \emph{splitter} $P$ which we identify with a Pauli string $P\in \lbrace{I, X, Y, Z\rbrace}^{V}$ in the obvious manner. The detachment of $F$ with respect to a splitter $P$ is the graph $F||P$ obtained after replacing each vertex $v\in W$ by two vertices $v'$ and $v''$, such that $v'$ is incident to the two half-edges of one pair, and $v''$ is incident to the two remaining half-edges.
    \end{definition}
See Fig.~\ref{fig:detachments_main} for a depiction of the three detachments at a vertex $v$.
\begin{figure}[h!]
    \centering
    \begin{tikzpicture}[
    scale=0.67,
    every node/.style={font=\small},
    vtx/.style={circle, fill=black, inner sep=1.9pt},
    hed/.style={line width=0.9pt}
]

\begin{scope}[xshift=4.4cm, yshift=2.8cm]
\filldraw[fill=gray!20, draw=black] (0,0) circle (1.3cm);
    \coordinate (N) at (0,1.2);
    \coordinate (E) at (1.2,0);
    \coordinate (S) at (0,-1.2);
    \coordinate (W) at (-1.2,0);

    \node[vtx] (v) at (0,0) {};
    \draw[hed] (v)--(N);
    \draw[hed] (v)--(E);
    \draw[hed] (v)--(S);
    \draw[hed] (v)--(W);
    \node[scale=1.3] at (0.3,0.3) {$v$};
\end{scope}

\begin{scope}[xshift=0cm, yshift=-1cm]
\filldraw[fill=gray!20, draw=black] (0,0) circle (1.3cm);
    \coordinate (N) at (0,1.2);
    \coordinate (E) at (1.2,0);
    \coordinate (S) at (0,-1.2);
    \coordinate (W) at (-1.2,0);

    \node[vtx] (x1) at (0.24,0.24) {};
    \node[vtx] (x2) at (-0.24,-0.24) {};

    \draw[hed] plot[smooth, tension=0.9] coordinates {(N) (x1) (E)};
    \draw[hed] plot[smooth, tension=0.9] coordinates {(S) (x2) (W)};

    \node[scale=1.5] at (0,1.75) {$X_v$};
\end{scope}

\begin{scope}[xshift=4.4cm, yshift=-1.3cm]
\filldraw[fill=gray!20, draw=black] (0,0) circle (1.3cm);
    \coordinate (N) at (0,1.2);
    \coordinate (E) at (1.2,0);
    \coordinate (S) at (0,-1.2);
    \coordinate (W) at (-1.2,0);

    \node[vtx] (y1) at (-0.30,0.26) {};

    \node[vtx] (y2) at (0.30,-0.26) {};

    \draw[hed] plot[smooth, tension=0.9] coordinates {(N) (y1) (S)};
    \filldraw[fill=gray!20, draw=none] (-0.25,-0.2) circle (0.1cm);
    
    \draw[hed] plot[smooth, tension=0.9] coordinates {(W) (y2) (E)};

    \node[scale=1.5] at (0,1.75) {$Y_v$};
\end{scope}

\begin{scope}[xshift=8.8cm, yshift=-1cm]
\filldraw[fill=gray!20, draw=black] (0,0) circle (1.3cm);
    \coordinate (N) at (0,1.2);
    \coordinate (E) at (1.2,0);
    \coordinate (S) at (0,-1.2);
    \coordinate (W) at (-1.2,0);

    \node[vtx] (z1) at (-0.24,0.24) {};
    \node[vtx] (z2) at (0.24,-0.24) {};

    \draw[hed] plot[smooth, tension=0.9] coordinates {(N) (z1) (W)};
    \draw[hed] plot[smooth, tension=0.9] coordinates {(E) (z2) (S)};

    \node[scale=1.5] at (0,1.75) {$Z_v$};
\end{scope}

\end{tikzpicture}
\caption{The three detachments at a vertex $v$. The under/overcrossing in the $Y_v$ detachment is only present for visual clarity. We have suppressed the labels of $v'$ and $v''$. 
}
\label{fig:detachments_main}
\end{figure}
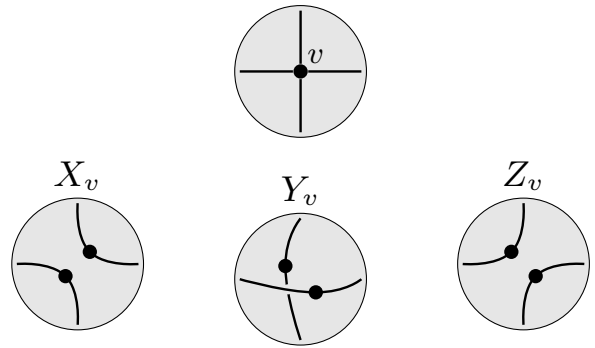

It is possible to associate a rank function to splitters on $4$-regular multigraphs.
\begin{definition}[informal]
The rank function associated with a $4$-regular multigraph $F$ on vertex set $V$ is a function $r_F:\lbrace{I, X, Y, Z\rbrace}^V\rightarrow\mathbb{N}_{\geq0}$ defined as
\begin{align}
r_F(P) = w(P) - c(F||P)+c(F) \ , \nonumber 
\end{align}
where $c(H)$ is the number of connected components of a graph $H$. 
\end{definition}
Rather surprisingly, Bouchet showed---using different terminology---that circle stabilizer states are exactly those states whose rank functions are described by rank functions of $4$-regular multigraphs~\cite{bouchet1988graphic, bouchet2001multimatroids}. We show an example of a $4$-regular multigraph corresponding to a 3-GHZ state in Fig.~\ref{fig:detachment_example}.

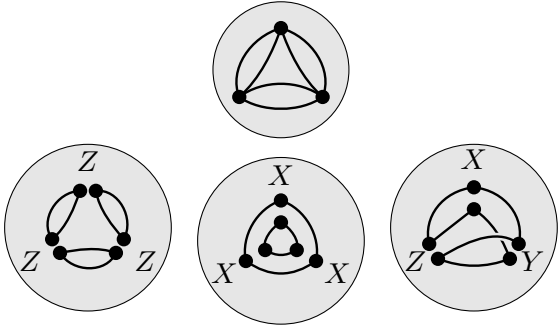
\begin{figure}[h!]
    \centering

\begin{tikzpicture}[
    scale=0.58,
    every node/.style={font=\small},
    vtx/.style={circle, fill=black, inner sep=1.9pt},
    hed/.style={line width=0.9pt, line cap=round, line join=round},
    arr/.style={-{Latex[length=2.7mm,width=2mm]}, line width=0.9pt}
]


\begin{scope}[xshift=4.4cm, yshift=3.2cm]
    \filldraw[fill=gray!20, draw=black] (0,0) circle (1.55cm);

    \coordinate (T) at (0,0.95);
    \coordinate (L) at (-0.95,-0.62);
    \coordinate (R) at (0.95,-0.62);

    \draw[hed] (T) .. controls (-0.70,0.75) and (-1.20,0.00) .. (L);
    \draw[hed] (T) .. controls (-0.18,0.42) and (-0.48,-0.20) .. (L);

    \draw[hed] (T) .. controls (0.70,0.75) and (1.20,0.00) .. (R);
    \draw[hed] (T) .. controls (0.18,0.42) and (0.48,-0.20) .. (R);

    \draw[hed] (L) .. controls (-0.42,-0.22) and (0.42,-0.22) .. (R);
    \draw[hed] (L) .. controls (-0.45,-0.98) and (0.45,-0.98) .. (R);

    \node[vtx] at (T) {};
    \node[vtx] at (L) {};
    \node[vtx] at (R) {};
\end{scope}





\begin{scope}[xshift=0cm, yshift=-0.4cm]
    \filldraw[fill=gray!20, draw=black] (0,0) circle (1.90cm);

    %
    \coordinate (topL) at (-0.18,0.84);
    \coordinate (topR) at ( 0.18,0.84);

    \foreach \ang in {0,120,240}{
        \begin{scope}[rotate=\ang]

            \coordinate (A) at (-0.18,0.84);
            \coordinate (B) at (-0.817,-0.264);

            \draw[hed]
                (A)
                .. controls (-0.72,0.82) and (-1.12,0.16) ..
                (B);

            \draw[hed]
                (A)
                .. controls (-0.31,0.46) and (-0.48,0.02) ..
                (B);

            \node[vtx] at (A) {};
            \node[vtx] at (B) {};

        \end{scope}
    }

    \node[scale=1.2] at ( 90:1.52) {$Z$};
    \node[scale=1.2] at (210:1.52) {$Z$};
    \node[scale=1.2] at (330:1.52) {$Z$};
\end{scope}


\begin{scope}[xshift=4.4cm, yshift=-0.7cm]
    \filldraw[fill=gray!20, draw=black] (0,0) circle (1.90cm);

    \coordinate (To) at (0,0.92);
    \coordinate (Lo) at (-0.80,-0.46);
    \coordinate (Ro) at (0.80,-0.46);

    \coordinate (Ti) at (0,0.42);
    \coordinate (Li) at (-0.36,-0.21);
    \coordinate (Ri) at (0.36,-0.21);

    \draw[hed] (To) .. controls (-0.55,0.74) and (-0.92,0.10) .. (Lo);
    \draw[hed] (Lo) .. controls (-0.25,-0.84) and (0.25,-0.84) .. (Ro);
    \draw[hed] (Ro) .. controls (0.92,0.10) and (0.55,0.74) .. (To);

    \draw[hed] (Ti) .. controls (-0.22,0.30) and (-0.44,-0.01) .. (Li);
    \draw[hed] (Li) .. controls (-0.10,-0.36) and (0.10,-0.36) .. (Ri);
    \draw[hed] (Ri) .. controls (0.44,-0.01) and (0.22,0.30) .. (Ti);

    \node[vtx] at (To) {};
    \node[vtx] at (Lo) {};
    \node[vtx] at (Ro) {};
    \node[vtx] at (Ti) {};
    \node[vtx] at (Li) {};
    \node[vtx] at (Ri) {};

    \node[scale=1.2] at ( 90:1.49) {$X$};
    \node[scale=1.2] at (210:1.49) {$X$};
    \node[scale=1.2] at (330:1.49) {$X$};
\end{scope}


\begin{scope}[xshift=8.8cm, yshift=-0.4cm]
    \filldraw[fill=gray!20, draw=black] (0,0) circle (1.90cm);

    \coordinate (Touter) at (0,0.92);
    \coordinate (Tinner) at (0,0.42);

    \coordinate (Ltop) at (-1.02,-0.37);
    \coordinate (Lbot) at (-0.82,-0.71);

    \coordinate (Rtop) at ( 1.02,-0.37);
    \coordinate (Rbot) at ( 0.82,-0.71);


    \draw[hed]
        (Touter)
        .. controls (-0.72,0.84) and (-1.17,0.23) ..
        (Ltop);

    \draw[hed]
        (Ltop)
        .. controls (-0.58,-0.02) and (-0.18,0.30) ..
        (Tinner);

    \draw[hed]
        (Tinner)
        .. controls (0.18,0.30) and (0.58,-0.02) ..
        (Rbot);

    \draw[hed]
        (Rbot)
        .. controls (0.26,-0.92) and (-0.26,-0.92) ..
        (Lbot);

    \draw[
        hed,
        preaction={draw=gray!20, line width=2.8pt}
    ]
        (Lbot)
        .. controls (0.02,-0.18) and (0.62,-0.06) ..
        (Rtop);

    \draw[hed]
        (Rtop)
        .. controls (1.17,0.23) and (0.72,0.84) ..
        (Touter);

    \node[vtx] at (Touter) {};
    \node[vtx] at (Tinner) {};
    \node[vtx] at (Ltop)   {};
    \node[vtx] at (Lbot)   {};
    \node[vtx] at (Rtop)   {};
    \node[vtx] at (Rbot)   {};

    \node[scale=1.2] at ( 90:1.55) {$X$};
    \node[scale=1.2] at (210:1.55) {$Z$};
    \node[scale=1.2] at (330:1.55) {$Y$};
\end{scope}

\end{tikzpicture}
\caption{A $4$-regular multigraph corresponding to the $3$-GHZ state in the top, and three possible detachments at the bottom. The bottom three multigraphs have $3$, $2$ and $1$ connected components, matching the expected evaluations of $r(ZZZ)=1$, $r(XXX)=2$ and $r(XYZ)=3$, respectively.
}
\label{fig:detachment_example}
\end{figure}

The connection between the rank function of a circle stabilizer state and a $4$-regular multigraph $F$ is convenient, since $d$ and $\alpha_{\rm loc}$ can be bounded/expressed in terms of the so-called girth and largest circuit decomposition of $F$, respectively. We show in Appendix \ref{sec:max_ent_in_circ} that these two quantities are necessarily constrained for $4$-regular multigraphs, leading to the following theorem.

\begin{theorem}\label{thm:main_circle_ent}
Circle stabilizer states satisfy
\begin{align}
d\left(\ket{\psi}\right)=&\,\,O\left(\log n\right),\nonumber \\
\alpha_{\rm loc}\left(\ket{\psi}\right) =&\,\,\Omega\left(\frac{n}{\log n }\right) .\nonumber 
\end{align}
\end{theorem}

\subsection{Geelen's weak structural conjecture}\label{sec:main_geelen}

In the previous section we showed that the entanglement in circle stabilizer states is limited, at least in the asymptotic sense. Here we argue that such statements hold true for all \emph{proper vertex-minor-closed families of stabilizer states}, conditioned on Geelen's weak structural conjecture and when restricting to states that are sufficiently \emph{rank-connected}.

Informally, Geelen's weak structural conjecture states that the entanglement of stabilizer states $\ket{\psi}$ in any \emph{proper vertex-minor-closed family} $\mathcal{F}$ is asymptotically \emph{close} to that of a subset of circle stabilizer states. Here, closeness is measured in terms of rank-$p$ perturbations (to be defined shortly). The set $\mathcal{F}$ being vertex-minor-closed means that for all $\ket{\psi}\in \mathcal{F}$, every vertex-minor $\ket{\phi}$ of $\ket{\psi}$ is in $\mathcal{F}$; proper means that $\mathcal{F}$ is a strict subset of the set of all stabilizer states~\cite{davies2025preparing, mccarty2021local}. As an example, for any set of states $\lbrace{\ket{\phi_i}\rbrace}$, the set of states not containing any of the $\ket{\phi_i}$ is a proper vertex-minor-closed family $\mathcal{F}$.

We now define the notion of a rank-$p$ perturbation~\cite{bourneuf2026combinatorial, mccarty2021local, campbell2026erdHos}.
\begin{definition}
    A \emph{rank-$p$ perturbation} of a graph $G$ (or graph state $\ket{G}$) is any graph whose adjacency matrix, up to diagonal elements, can be obtained by adding a symmetric matrix of rank at most $p$ to the adjacency matrix of $G$.
\end{definition}

Geelen's conjecture also requires the notions of $k$-rank-connectedness.

\begin{definition}
    A graph $G$ on vertex set $V$ is $k$-rank-connected if $\left|V\right|\geq 2k$ and $\textrm{cutrk}(X) \geq \min(\left|X\right|, \left|V\setminus X\right|, k)$ for all $X\subseteq V$.
\end{definition}

We are now ready to state Geelen's weak conjecture.
\begin{conjecture}[Geelen's weak vertex-minor structure conjecture]\label{conj:geelen_main_text}
For every class of stabilizer states $\mathcal{F}$ which is proper vertex-minor-closed, there exist $k, p\in \mathbb{N}$ such that each $\ket{\psi}\in \mathcal{F}$ that is $k$-rank-connected is single-qubit Clifford equivalent to a rank-$p$ perturbation of a circle graph state.
\end{conjecture}

We note that this conjecture is normally phrased purely in terms of graphs instead of stabilizer states; these statements are equivalent, however.

In Appendix \ref{sec:ent_under_rank_perturbations} we show that, for a fixed $p$, rank-$p$ perturbations can only change the $\alpha_{\rm loc}$ measure by a constant additive factor. A slightly more involved statement can be made for the distance. Combined with theorem \ref{thm:main_circle_ent}, our main result from theorem \ref{thm:main_result_informal} follows.

\section{Discussion}
We have shown (conditional on Geelen's weak structural conjecture) that excluding a vertex-minor necessarily limits entanglement in the asymptotic setting.
While our results hold in the case of sufficiently rank-connected states, we believe it is possible to relax this condition by using Geelen's full structural conjecture~\cite{mccarty2021local}. 

Our results rely on finding a direct connection between the probability distribution of measurement outcomes under single-qubit Pauli measurements with the rank functions of multimatroids. There is a rich literature on multimatroids~\cite{bouchet1997multimatroids, bouchet1998multimatroids, bouchet2001multimatroids, bouchet19984multimatroids, noble2025primer}, and in particular the subclass corresponding to stabilizer states~\cite{ bouchet1987isotropic, brijder2015isotropic1, bouchet1989connectivity, brijder2015isotropic, traldi2016isotropic}; further deepening this connection may provide a systematic way to translate results from multimatroid theory into multipartite stabilizer entanglement and vice versa.

Two natural generalizations present themselves. First, the notion of a vertex-minor can be extended beyond the stabilizer formalism. That is, a (not necessarily stabilizer) state $\ket{\psi}$ contains a state $\ket{\phi}$ if $\ket{\phi}$ can be reached through (stochastic) LOCC operations. Do similar statements as theorem \ref{thm:main_result_informal} hold as well in this setting? Similarly, even when restricting to stabilizer states, one can consider arbitray local unitaries, instead of local Cliffords. Do the statements of this manuscript qualitatively change when allowing for local unitaries? This connects closely with the failure of the LU=LC conjecture to hold~\cite{ji2007lu, claudet2024local, claudet2025local, claudet202627}.

Second, stabilizer states can be generalized to stabilizer codes~\cite{gottesman1997stabilizer, khesin2025universal,  goodenough2024bipartite}. As with stabilizer states, stabilizer codes have a notion of distance that quantifies their error correction capabilities. In practice, a linear distance is desired, so a similar limitation on the distance when forbidding a vertex-minor would impose strong no-go statements on large-scale quantum computation.

Finally, the proof techniques used in this work might provide a way to tackle the McCarty-Geelen conjecture, which states that measurement-based quantum computation can be efficiently classically simulated on states in proper vertex-minor-closed families~\cite{mccarty2021local,harrison2025fermionic}.

{\it Acknowledgments---}%
We thank Axel Dahlberg, Vicente Lenz, Vis Taraz and David Elkouss for discussions.
KG acknowledges the support from the Alexander von Humboldt Foundation. Figures were created with ChatGPT.

\clearpage
\newpage

\bibliography{references}

@inproceedings{goodenough2024bipartite,
  title={Bipartite entanglement of noisy stabilizer states through the lens of stabilizer codes},
  author={Goodenough, Kenneth and Sajjad, Aqil and Kaur, Eneet and Guha, Saikat and Towsley, Don},
  booktitle={2024 IEEE International Symposium on Information Theory (ISIT)},
  pages={545--550},
  year={2024},
  organization={IEEE}
}

@article{hein2004multiparty,
  title={Multiparty entanglement in graph states},
  author={Hein, Marc and Eisert, Jens and Briegel, Hans J},
  journal={Physical Review A—Atomic, Molecular, and Optical Physics},
  volume={69},
  number={6},
  pages={062311},
  year={2004},
  publisher={APS}
}

@article{goodenough2024near,
  title={Near-term n to k distillation protocols using graph codes},
  author={Goodenough, Kenneth and De Bone, Sebastian and Addala, Vaishnavi and Krastanov, Stefan and Jansen, Sarah and Gijswijt, Dion and Elkouss, David},
  journal={IEEE Journal on Selected Areas in Communications},
  volume={42},
  number={7},
  pages={1830--1849},
  year={2024},
  publisher={IEEE}
}

@article{dahlberg2022complexity,
  title={The complexity of the vertex-minor problem},
  author={Dahlberg, Axel and Helsen, Jonas and Wehner, Stephanie},
  journal={Information Processing Letters},
  volume={175},
  pages={106222},
  year={2022},
  publisher={Elsevier}
}

@article{dahlberg2018transforming,
  title={Transforming graph states using single-qubit operations},
  author={Dahlberg, Axel and Wehner, Stephanie},
  journal={Philosophical Transactions of the Royal Society A: Mathematical, Physical and Engineering Sciences},
  volume={376},
  number={2123},
  pages={20170325},
  year={2018},
  publisher={The Royal Society Publishing}
}

@article{bouchet1994circle,
  title={Circle graph obstructions},
  author={Bouchet, Andr{\'e}},
  journal={Journal of Combinatorial Theory, Series B},
  volume={60},
  number={1},
  pages={107--144},
  year={1994},
  publisher={Elsevier}
}

@article{brijder2015isotropic,
  title={Isotropic matroids {II}: Circle graphs},
  author={Brijder, Robert and Traldi, Lorenzo},
  journal={arXiv preprint arXiv:1504.04299},
  year={2015}
}

@article{alon2002moore,
  title={The {M}oore bound for irregular graphs},
  author={Alon, Noga and Hoory, Shlomo and Linial, Nathan},
  journal={Graphs and Combinatorics},
  volume={18},
  number={1},
  pages={53--57},
  year={2002},
  publisher={Springer}
}

@article{davies2025preparing,
  title={Preparing graph states forbidding a vertex-minor},
  author={Davies, James and Jena, Andrew},
  journal={arXiv preprint arXiv:2504.00291},
  year={2025}
}

@phdthesis{mccarty2021local,
  title={Local structure for vertex-minors},
  author={McCarty, Rose},
  year={2021},
  school={University of Waterloo}
}

@article{markham2007entanglement,
  title={Entanglement and local information access for graph states},
  author={Markham, Damian and Miyake, Akimasa and Virmani, Shashank},
  journal={New Journal of Physics},
  volume={9},
  number={6},
  pages={194},
  year={2007},
  publisher={IOP Publishing}
}

@article{weinbrenner2025quantifying,
  title={Quantifying entanglement from the geometric perspective},
  author={Weinbrenner, Lisa T and G{\"u}hne, Otfried},
  journal={arXiv preprint arXiv:2505.01394},
  year={2025}
}

@article{bouchet1998multimatroids,
  title={Multimatroids {II}. {Orthogonality}, minors and connectivity},
  author={Bouchet, Andr{\'e}},
  journal={the electronic journal of combinatorics},
  pages={R8--R8},
  year={1998}
}

@article{brijder2015isotropic1,
  title={Isotropic matroids {I}: {Multimatroids} and neighborhoods},
  author={Brijder, Robert and Traldi, Lorenzo},
  journal={arXiv preprint arXiv:1503.04406},
  year={2015}
}

@article{bouchet1997multimatroids,
  title={Multimatroids {I}. {Coverings} by independent sets},
  author={Bouchet, Andr{\'e}},
  journal={SIAM Journal on Discrete Mathematics},
  volume={10},
  number={4},
  pages={626--646},
  year={1997},
  publisher={SIAM}
}

@article{niekamp2012entropic,
  title={Entropic uncertainty relations and the stabilizer formalism},
  author={Niekamp, S{\"o}nke and Kleinmann, Matthias and G{\"u}hne, Otfried},
  journal={Journal of mathematical physics},
  volume={53},
  number={1},
  year={2012},
  publisher={AIP Publishing}
}

@inproceedings{javelle2012minimum,
  title={On the minimum degree up to local complementation: Bounds and complexity},
  author={Javelle, J{\'e}r{\^o}me and Mhalla, Mehdi and Perdrix, Simon},
  booktitle={International Workshop on Graph-Theoretic Concepts in Computer Science},
  pages={138--147},
  year={2012},
  organization={Springer}
}

@article{hein2006entanglement,
  title={Entanglement in graph states and its applications},
  author={Hein, Marc and D{\"u}r, Wolfgang and Eisert, Jens and Raussendorf, Robert and Nest, M and Briegel, H-J},
  journal={arXiv preprint quant-ph/0602096},
  year={2006}
}

@incollection{noble2025primer,
  title={A Primer on Delta-matroids and Multimatroids},
  author={Noble, Steven},
  booktitle={2023 MATRIX Annals},
  pages={391--420},
  year={2025},
  publisher={Springer}
}

@article{merino2025activities,
  title={An activities expansion of the transition polynomial of a multimatroid},
  author={Merino, Criel and Moffatt, Iain and Noble, Steven},
  journal={SIAM Journal on Discrete Mathematics},
  volume={39},
  number={2},
  pages={1372--1407},
  year={2025},
  publisher={SIAM}
}

@article{brijder2014interlace,
  title={Interlace polynomials for multimatroids and delta-matroids},
  author={Brijder, Robert and Hoogeboom, Hendrik Jan},
  journal={European Journal of Combinatorics},
  volume={40},
  pages={142--167},
  year={2014},
  publisher={Elsevier}
}

@article{harrison2025fermionic,
  title={Fermionic Insights into Measurement-Based Quantum Computation: Circle Graph States Are Not Universal Resources},
  author={Harrison, Brent and Iyer, Vishnu and Parekh, Ojas and Thompson, Kevin and Zhao, Andrew},
  journal={arXiv preprint arXiv:2510.05557},
  year={2025}
}

@article{geelen2023grid,
  title={The grid theorem for vertex-minors},
  author={Geelen, Jim and Kwon, O-joung and McCarty, Rose and Wollan, Paul},
  journal={Journal of Combinatorial Theory, Series B},
  volume={158},
  pages={93--116},
  year={2023},
  publisher={Elsevier}
}

@article{oum2017rank,
  title={Rank-width: Algorithmic and structural results},
  author={Oum, Sang-il},
  journal={Discrete Applied Mathematics},
  volume={231},
  pages={15--24},
  year={2017},
  publisher={Elsevier}
}

@book{gottesman1997stabilizer,
  title={Stabilizer codes and quantum error correction},
  author={Gottesman, Daniel},
  year={1997},
  publisher={California Institute of Technology}
}

@article{shettell2020graph,
  title={Graph states as a resource for quantum metrology},
  author={Shettell, Nathan and Markham, Damian},
  journal={Physical review letters},
  volume={124},
  number={11},
  pages={110502},
  year={2020},
  publisher={APS}
}

@article{markham2008graph,
  title={Graph states for quantum secret sharing},
  author={Markham, Damian and Sanders, Barry C},
  journal={Physical Review A—Atomic, Molecular, and Optical Physics},
  volume={78},
  number={4},
  pages={042309},
  year={2008},
  publisher={APS}
}

@article{claudet2024local,
  title={Local equivalence of stabilizer states: a graphical characterisation},
  author={Claudet, Nathan and Perdrix, Simon},
  journal={arXiv preprint arXiv:2409.20183},
  year={2024}
}

@article{dahlberg2020counting,
  title={Counting single-qubit {Clifford} equivalent graph states is \#{P}-complete},
  author={Dahlberg, Axel and Helsen, Jonas and Wehner, Stephanie},
  journal={Journal of Mathematical Physics},
  volume={61},
  number={2},
  year={2020},
  publisher={AIP Publishing}
}

@article{van2004graphical,
  title={Graphical description of the action of local {Clifford} transformations on graph states},
  author={Van den Nest, Maarten and Dehaene, Jeroen and De Moor, Bart},
  journal={Physical Review A},
  volume={69},
  number={2},
  pages={022316},
  year={2004},
  publisher={APS}
}

@article{van2007classical,
  title={Classical simulation versus universality in measurement-based quantum computation},
  author={Van den Nest, Maarten and D{\"u}r, Wolfgang and Vidal, Guifr{\'e} and Briegel, Hans J},
  journal={Physical Review A—Atomic, Molecular, and Optical Physics},
  volume={75},
  number={1},
  pages={012337},
  year={2007},
  publisher={APS}
}

@article{raissi2022general,
  title={General stabilizer approach for constructing highly entangled graph states},
  author={Raissi, Zahra and Burchardt, Adam and Barnes, Edwin},
  journal={Physical Review A},
  volume={106},
  number={6},
  pages={062424},
  year={2022},
  publisher={APS}
}

@article{briegel2001persistent,
  title={Persistent entanglement in arrays of interacting particles},
  author={Briegel, Hans J and Raussendorf, Robert},
  journal={Physical Review Letters},
  volume={86},
  number={5},
  pages={910},
  year={2001},
  publisher={APS}
}

@article{prielinger2025piecemaker,
  title={Piecemaker: a resource-efficient entanglement distribution protocol},
  author={Prielinger, Luise and Goodenough, Kenneth and Avis, Guus and Krastanov, Stefan and Towsley, Don and Vardoyan, Gayane},
  journal={arXiv preprint arXiv:2508.14737},
  year={2025}
}

@inproceedings{grassl2002graphs,
  title={Graphs, quadratic forms, and quantum codes},
  author={Grassl, Markus and Klappenecker, Andreas and Rotteler, Martin},
  booktitle={Proceedings IEEE International Symposium on Information Theory,},
  pages={45},
  year={2002},
  organization={IEEE}
}

@article{khesin2025universal,
  title={Universal graph representation of stabilizer codes},
  author={Khesin, Andrey Boris and Lu, Jonathan Z and Shor, Peter W},
  journal={PRX Quantum},
  volume={6},
  number={4},
  pages={040325},
  year={2025},
  publisher={APS}
}

@article{li2022photonic,
  title={Photonic resource state generation from a minimal number of quantum emitters},
  author={Li, Bikun and Economou, Sophia E and Barnes, Edwin},
  journal={npj Quantum Information},
  volume={8},
  number={1},
  pages={11},
  year={2022},
  publisher={Nature Publishing Group UK London},
nolink={}
}

@incollection {cattaneo2015mindeg,
    AUTHOR = {Cattan\'eo, David and Perdrix, Simon},
     TITLE = {Minimum degree up to local complementation: bounds,
              parameterized complexity, and exact algorithms},
 BOOKTITLE = {Algorithms and computation},
    SERIES = {Lecture Notes in Comput. Sci.},
    VOLUME = {9472},
     PAGES = {259--270},
 PUBLISHER = {Springer, Heidelberg},
      YEAR = {2015},
      ISBN = {978-3-662-48971-0; 978-3-662-48970-3},
   MRCLASS = {68R10 (68Q25)},
  MRNUMBER = {3489490},
       DOI = {10.1007/978-3-662-48971-0\_23},
       URL = {https://doi.org/10.1007/978-3-662-48971-0_23},
}

@article{bhatti2025distributing,
  title={Distributing graph states with a photon-weaving quantum server},
  author={Bhatti, Daniel and Goodenough, Kenneth},
  journal={arXiv preprint arXiv:2504.07410},
  year={2025}
}

@article{traldi2011binary,
  title={Binary nullity, {Euler} circuits and interlace polynomials},
  author={Traldi, Lorenzo},
  journal={European Journal of Combinatorics},
  volume={32},
  number={6},
  pages={944--950},
  year={2011},
  publisher={Elsevier}
}

@article{oum2023rank,
  title={Rank connectivity and pivot-minors of graphs},
  author={Oum, Sang-il},
  journal={European Journal of Combinatorics},
  volume={108},
  pages={103634},
  year={2023},
  publisher={Elsevier}
}

@article{abramsky2017complete,
  title={A complete characterization of all-versus-nothing arguments for stabilizer states},
  author={Abramsky, Samson and Barbosa, Rui Soares and Car{\`u}, Giovanni and Perdrix, Simon},
  journal={Philosophical Transactions of the Royal Society A: Mathematical, Physical and Engineering Sciences},
  volume={375},
  number={2106},
  pages={20160385},
  year={2017},
  publisher={The Royal Society Publishing}
}

@article{weinbrenner2026complete,
  title={Complete Hierarchies for the Geometric Measure of Entanglement},
  author={Weinbrenner, Lisa T and Rico, Albert and Goodenough, Kenneth and Yu, Xiao-Dong and G{\"u}hne, Otfried},
  journal={arXiv preprint arXiv:2601.23243},
  year={2026}
}

@article{bouchet2001multimatroids,
  title={Multimatroids {III}. {Tightness} and fundamental graphs},
  author={Bouchet, Andr{\'e}},
  journal={European Journal of Combinatorics},
  volume={22},
  number={5},
  pages={657--677},
  year={2001},
  publisher={Elsevier}
}

@article{de2011linearized,
  title={A linearized stabilizer formalism for systems of finite dimension},
  author={De Beaudrap, Niel},
  journal={arXiv preprint arXiv:1102.3354},
  year={2011}
}

@article{bouchet1993compatible,
  title={Compatible {Euler} tours and supplementary {Eulerian} vectors},
  author={Bouchet, Andr{\'e}},
  journal={European journal of combinatorics},
  volume={14},
  number={6},
  pages={513--520},
  year={1993},
  publisher={Elsevier}
}

@article{jackson1991supplementary,
  title={Supplementary {Eulerian} vectors in isotropic systems},
  author={Jackson, Bill},
  journal={Journal of Combinatorial Theory Series B},
  volume={53},
  number={1},
  pages={93--105},
  year={1991},
  publisher={Academic Press, Inc. Orlando, FL, USA}
}

@article{bouchet1987isotropic,
  title={Isotropic systems},
  author={Bouchet, Andr{\'e}},
  journal={European Journal of Combinatorics},
  volume={8},
  number={3},
  pages={231--244},
  year={1987},
  publisher={Elsevier}
}

@article{bouchet1988graphic,
  title={Graphic presentations of isotropic systems},
  author={Bouchet, Andr{\'e}},
  journal={Journal of Combinatorial Theory, Series B},
  volume={45},
  number={1},
  pages={58--76},
  year={1988},
  publisher={Elsevier}
}

@inproceedings{bouchet1989connectivity,
  title={Connectivity of isotropic systems},
  author={Bouchet, Andr{\'e}},
  booktitle={Proceedings of the third international conference on Combinatorial mathematics},
  pages={81--93},
  year={1989}
}

@article{aaronson2004improved,
  title={Improved simulation of stabilizer circuits},
  author={Aaronson, Scott and Gottesman, Daniel},
  journal={Physical Review A—Atomic, Molecular, and Optical Physics},
  volume={70},
  number={5},
  pages={052328},
  year={2004},
  publisher={APS}
}

@article{brandhofer2025hardware,
  title={Hardware-efficient preparation of architecture-specific graph states on near-term quantum computers},
  author={Brandhofer, Sebastian and Polian, Ilia and Barz, Stefanie and Bhatti, Daniel},
  journal={Scientific Reports},
  volume={15},
  number={1},
  pages={2095},
  year={2025},
  publisher={Nature Publishing Group UK London}
}

@article{bouchet19984multimatroids,
  title={Multimatroids {IV}. {Chain-group} representations},
  author={Bouchet, Andr{\'e}},
  journal={Linear algebra and its applications},
  volume={277},
  number={1-3},
  pages={271--289},
  year={1998},
  publisher={Elsevier}
}

@inproceedings{kloks1993treewidth,
  title={Treewidth of circle graphs},
  author={Kloks, Ton},
  booktitle={International Symposium on Algorithms and Computation},
  pages={108--117},
  year={1993},
  organization={Springer}
}

@article{davies2019circle,
  title={Circle graphs are quadratically $\chi $-bounded},
  author={Davies, James and McCarty, Rose},
  journal={arXiv preprint arXiv:1905.11578},
  year={2019}
}

@inproceedings{campbell2026erdHos,
  title={The Erd{\H{o}}s-P{\'o}sa property for circle graphs as vertex-minors},
  author={Campbell, Rutger and Gollin, J Pascal and Hatzel, Meike and Kwon, O-joung and McCarty, Rose and Oum, Sang-il and Wiederrecht, Sebastian},
  booktitle={Proceedings of the 2026 Annual ACM-SIAM Symposium on Discrete Algorithms (SODA)},
  pages={4930--4952},
  year={2026},
  organization={SIAM}
}

@article{hahn2026structure,
  title={The Structure of Circle Graph States},
  author={Hahn, Frederik and McCarty, Rose and Nautrup, Hendrik Poulsen and Claudet, Nathan},
  journal={arXiv preprint arXiv:2603.08847},
  year={2026}
}

@article{claudet2025local,
  title={Local Equivalences of Graph States},
  author={Claudet, Nathan},
  journal={arXiv preprint arXiv:2511.22271},
  year={2025}
}

@article{ji2007lu,
  title={The {LU-LC} conjecture is false},
  author={Ji, Zhengfeng and Chen, Jianxin and Wei, Zhaohui and Ying, Mingsheng},
  journal={arXiv preprint arXiv:0709.1266},
  year={2007}
}

@article{claudet202627,
  title={The 27-qubit Counterexample to the {LU-LC} Conjecture is Minimal},
  author={Claudet, Nathan},
  journal={arXiv preprint arXiv:2603.25219},
  year={2026}
}

@article{traldi2016isotropic,
  title={Isotropic matroids {III}: {C}onnectivity},
  author={Traldi, Lorenzo and Brijder, Robert},
  journal={arXiv preprint arXiv:1602.03899},
  year={2016}
}

@book{welsh2010matroid,
  title={Matroid theory},
  author={Welsh, Dominic JA},
  year={2010},
  publisher={Courier Corporation}
}

@article{brijder2018orienting,
  title={Orienting transversals and transition polynomials of multimatroids},
  author={Brijder, Robert},
  journal={Advances in Applied Mathematics},
  volume={94},
  pages={120--155},
  year={2018},
  publisher={Elsevier}
}

@article{ascoli2026almost,
  title={Almost all graphs are vertex-minor universal},
  author={Ascoli, Ruben and Frederickson, Bryce and Frederickson, Sarah and McFarland, Caleb and Post, Logan},
  journal={arXiv preprint arXiv:2602.09049},
  year={2026}
}

@article{bae2026vertex,
  title={Vertex-minor {R}amsey numbers: exact values and extremal structure},
  author={Bae, Ji Ho},
  journal={arXiv preprint arXiv:2604.13434},
  year={2026}
}

@book{graham1991ramsey,
  title={Ramsey theory},
  author={Graham, Ronald L and Rothschild, Bruce L and Spencer, Joel H},
  year={1991},
  publisher={John Wiley \& Sons}
}

@article{geelen2008some,
  title={Some open problems on excluding a uniform matroid},
  author={Geelen, Jim},
  journal={Advances in applied mathematics},
  volume={41},
  number={4},
  pages={628--637},
  year={2008},
  publisher={Elsevier}
}

@book{nielsen2010quantum,
  title={Quantum computation and quantum information},
  author={Nielsen, Michael A and Chuang, Isaac L},
  year={2010},
  publisher={Cambridge university press}
}

@article{eisert2008area,
  title={Area laws for the entanglement entropy-a review},
  author={Eisert, Jens and Cramer, Marcus and Plenio, Martin B},
  journal={arXiv preprint arXiv:0808.3773},
  year={2008}
}

@article{tzitrin2018local,
  title={Local equivalence of complete bipartite and repeater graph states},
  author={Tzitrin, Ilan},
  journal={Physical Review A},
  volume={98},
  number={3},
  pages={032305},
  year={2018},
  publisher={APS}
}

@article{cabello2009entanglement,
  title={Entanglement in eight-qubit graph states},
  author={Cabello, Ad{\'a}n and L{\'o}pez-Tarrida, Antonio J and Moreno, Pilar and Portillo, Jos{\'e} R},
  journal={Physics Letters A},
  volume={373},
  number={26},
  pages={2219--2225},
  year={2009},
  publisher={Elsevier}
}

@article{fattal2004entanglement,
  title={Entanglement in the stabilizer formalism},
  author={Fattal, David and Cubitt, Toby S and Yamamoto, Yoshihisa and Bravyi, Sergey and Chuang, Isaac L},
  journal={arXiv preprint quant-ph/0406168},
  year={2004}
}

@incollection{GreenbergerHorneZeilinger1989,
  title     = {Going Beyond Bell's Theorem},
  author    = {Greenberger, Daniel M. and Horne, Michael A. and Zeilinger, Anton},
  booktitle = {Bell's Theorem, Quantum Theory and Conceptions of the Universe},
  editor    = {Kafatos, Menas},
  publisher = {Kluwer Academic},
  address   = {Dordrecht},
  pages     = {69--72},
  year      = {1989}
}

@article{brijder2022characterization,
  title={A characterization of circle graphs in terms of total unimodularity},
  author={Brijder, Robert and Traldi, Lorenzo},
  journal={European Journal of Combinatorics},
  volume={102},
  pages={103455},
  year={2022},
  publisher={Elsevier}
}

@article{bourneuf2026combinatorial,
  title={A combinatorial framework for clustering graph states: Algorithms and hardness for rank-integrity},
  author={Bourneuf, Romain and Claudet, Nathan and Kim, Sang Yoon and McCarty, Rose and Sullivan, Blair D and Thomass{\'e}, St{\'e}phan},
  journal={arXiv preprint arXiv:2607.09469},
  year={2026}
}

@article{goodenough2026exact,
  title={Exact noise characterization of entanglement distribution in star networks},
  author={Goodenough, Kenneth and Chen, Xiaonan and Emonts, Patrick},
  journal={arXiv preprint arXiv:2606.07043},
  year={2026}
}

@article{wei1991generalized,
  title={Generalized Hamming weights for linear codes},
  author={Wei, Victor K},
  journal={IEEE Transactions on information theory},
  volume={37},
  number={5},
  pages={1412--1418},
  year={1991},
  publisher={IEEE}
}

\newpage

\onecolumngrid
\section*{APPENDIX}
\appendix


We sketch here an outline of each section.\\

\noindent \paragraph*{Section~\ref{sec:prelim} (Basics of the stabilizer formalism).}
We briefly review the stabilizer formalism, defining stabilizer groups, graph states, local equivalence and vertex-minors. We extend notions commonly defined only for graph states to stabilizer states.\\

\noindent \paragraph*{Section~\ref{sec:single_qubit} (Single-qubit measurements and rank/nullity functions).}
We develop a general framework for the statistics of single-qubit Pauli measurements on stabilizer states. In particular, we show that the so-called rank and nullity of a Pauli measurement determine the randomness/information in the outcomes of such a Pauli measurement. We then reformulate these notions in terms of isotropic matroids, which will be convenient for some of the proofs.\\

\noindent \paragraph*{Section~\ref{sec:ent_measures} (Entanglement measures).}
We show that the distance $d$ and the locally accessible information $\alpha_{\rm loc}$ are naturally captured by the rank/nullity. We interpret these quantities in terms of how much information can be extracted using single-qubit Pauli measurements: $d$ corresponds to the smallest support from which any information can be obtained, while $\alpha_{\rm loc}$ captures the maximum number of extractable bits.\\

\noindent \paragraph*{Section~\ref{sec:single_qubit_on_circle} (Single-qubit Pauli measurements on circle stabilizer states).}
We specialize the above measurement framework to circle graph states/circle stabilizer states. Bouchet \cite{bouchet2001multimatroids} defined a rank function on $4$-regular multigraphs, and showed---using different terminology--- that the rank function of a stabilizer state $\ket{\psi}$ can be written as that of the \emph{rank function of a $4$-regular multigraph} iff $\ket{\psi}$ is a circle stabilizer state~\cite{bouchet1988graphic}. The rank/nullity on $4$-regular multigraphs can be expressed in terms of detachments and circuit partitions of the underlying multigraph. This yields characterizations of the distance and the locally accessible information of circle stabilizer states in terms of shortest cycles and largest cycle partitions, respectively.\\

\noindent \paragraph*{Section~\ref{sec:max_ent_in_circ} (Maximum entanglement in circle stabilizer states).}
We use properties of $4$-regular multigraph to show that the entanglement of circle stabilizer states is constrained in the asymptotic limit. In particular, we show that the distance and locally accessible information satisfy $d=O(\log n)$ and $\alpha_{\rm loc}=\Omega\!\left(\frac{n}{\log n}\right)$ for circle stabilizer states.\\

\noindent \paragraph*{Section~\ref{sec:geelens_conj} (PVMC families are asymptotically weakly entangled).}
We first discuss the basics of vertex-minor theory, such as local equivalence and vertex-minors.
We then discuss the notions of proper vertex-minor-closed families, rank-$p$ perturbations and ($k$-)rank-connectivity, after which we have the background to state Geelen's weak vertex-minor structure conjecture.
This will allow us to `lift' the constraints on the asymptotic entanglement in circle stabilizer states to arbitrary proper vertex-minor-closed families (conditional on Geelen's weak structural conjecture and the states being sufficiently rank-connected).\\

\section{Basics of the stabilizer formalism}\label{sec:prelim}

In this section we collect some of the basic background and notation used throughout the paper. We begin with a brief review of the stabilizer formalism and graph states, after which we recall the cut-rank function and its connection to bipartite entanglement in graph/stabilizer states. For a more thorough introduction, see for example~\cite{gottesman1997stabilizer} or~\cite{nielsen2010quantum}. 

Afterwards, we give the relevant notions of local equivalence and vertex-minors. For posterity, we will extend here definitions from graphs/graph states to stabilizer states whenever possible.

\subsection{Stabilizer formalism}\label{sec:formalism}

The stabilizer formalism provides an efficient algebraic description of a physically relevant class of multipartite quantum states. Instead of specifying a state vector explicitly in a $2^n$-dimensional complex Hilbert space over $\mathbb{C}$, one characterizes the state as the common $+1$ eigenspace of certain abelian subgroups of the $n$-qubit Pauli group. This description is significantly more compact, especially when restricting to Clifford and Pauli measurements as we will do.

We work with $n$ qubits, whose state space is $(\mathbb{C}^2)^{\otimes n}$. On one qubit,
the Pauli matrices are
\begin{align*}
I=\begin{pmatrix}1&0\\0&1\end{pmatrix},\qquad
X=\begin{pmatrix}0&1\\1&0\end{pmatrix},\qquad
Y=\begin{pmatrix}0&-i\\i&0\end{pmatrix},\qquad
Z=\begin{pmatrix}1&0\\0&-1\end{pmatrix}.
\end{align*}
A \emph{Pauli string} is a tensor product
\begin{align*}
P=P_1\otimes \cdots \otimes P_n,
\qquad P_j\in \{I,X,Y,Z\}.
\end{align*}
That is, a Pauli string applies one of the four matrices $I,X,Y,Z$ independently on each
qubit. For the sake of brevity and convenience, we will often omit the tensor sign and write the Pauli string just as strings of elements in $\{I,X,Y,Z\}$. For instance, for $n=3$ the element $X\otimes X\otimes Z$, will be just denoted by $XXZ$. We shall denote the identity $III$ by $I^{\otimes 3}$, or just $I$ when $n$ is clear from the context. We start with the definition of a Pauli group.

\begin{definition}[Pauli group]
The $n$-qubit \emph{Pauli group} $\mathcal{P}_n$ is defined as the group consisting of all Pauli strings of length $n$ together with their phases, i.e.,
\begin{align}
\mathcal P_n= \left\{\omega P_1\otimes\cdots\otimes P_n:
\omega\in\{\pm 1,\pm i\},\; P_j\in\{I,X,Y,Z\} \right\}.
\end{align}
\end{definition}

As we are going to see in later sections, even though we refer to the phases of the elements of the group, our results do not require the tracking of these phases.

\begin{definition}[Support and weight]
    The \emph{support} $\textrm{supp}(P)$ of a Pauli string $P$ is the set of qubits it acts on non-trivially, i.e.,
    \begin{align*}
    \operatorname{supp}(P)=\{j\in[n]:P_j\neq I\}.
    \end{align*}
    The \emph{weight} $w(P)$ of a Pauli string $P$ equals $\left|\textrm{supp}(P)\right|$.
\end{definition}

Let $\mathcal{S}\subseteq \mathcal{P}_n$ be a subgroup of the Pauli group. We say that $\mathcal S$ is generated by $S_1,\ldots,S_r$ if every element of $\mathcal S$ is a product of these generators. The generators are called \emph{algebraically independent} if no non-empty product of distinct elements is equal to $I$.

\begin{definition}[Stabilizer states/groups]
A \emph{stabilizer group} $\mathcal{S}\subset \mathcal{P}_n$ is an abelian subgroup of the Pauli group not containing the element $-I$, and generated by $n$ algebraically independent commuting Pauli operators. A pure state $\ket{\psi}$ is a \emph{stabilizer state} if it is the common $+1$ eigenstate of all elements of $\mathcal{S}$, i.e.
\begin{align*}
S\ket{\psi}=\ket{\psi}
\end{align*}
for every $S\in \mathcal{S}$.
\end{definition}

We remark that the state $\ket{\psi}$ is unique up to global phases. Let us briefly justify this assertion for completeness. Let
$S_1,\ldots,S_n$ be independent commuting generators of $\mathcal S$. Since $-I$ is not
in $\mathcal S$, each $S_i$ satisfies $S_i^2=I$. Moreover, it holds that $S_i$ has eigenvalues $\pm 1$. This implies that the operator $\prod_{i=1}^n\frac{I+S_i}{2}$ is a projector onto the common $+1$ eigenspace of $\mathcal{S}$. This subspace has dimension one, because
of the fact that $\prod_{i=1}^n\frac{I+S_i}{2}$ has eigenvalues $0$ and $1$, and that $\mathrm{Tr}(P)=0$ for all Pauli strings, except for $\mathrm{Tr}(I)=2^n$. From this it follows that the dimension of the projector $\prod_{i=1}^n\frac{I+S_i}{2}$ equals $\mathrm{Tr}\left(\prod_{i=1}^n\frac{I+S_i}{2}\right)=1$. The state $\ket{\psi}$ is then the unique unit norm vector (up to global phase) in that space. Finally, it holds that
\begin{align}\label{eq:stab_dm}
\ket{\psi}\!\bra{\psi}
= \prod_{i=1}^n\frac{I+S_i}{2}=
\frac{1}{2^n}\sum_{S\in\mathcal S} S.
\end{align}

For a subgroup $\mathcal{S}$ of the Pauli group (which may not be a stabilizer group), we write $\textrm{dim}\left(\mathcal{S}\right)$ for the size of a minimum generating set. In the case that $\mathcal{S}$ is abelian and does not contain $-I$, one can verify that $\textrm{dim}\left(\mathcal{S}\right)=\log_2\left(\left|\mathcal{S}\right|\right)$.

\medskip
A particularly important subclass of stabilizer states is given by graph states.

\begin{definition}[Graph states]
Let $G=(V,E)$ be a simple graph on $n$ vertices. The \emph{graph state} $\ket{G}$ is the stabilizer state whose stabilizer group is generated by the algebraically independent and commuting \emph{vertex-stabilizers}
\[
S_v = X_v \prod_{u\in N_v} Z_u,
\qquad v\in V,
\]
where $N_v$ denotes the open neighborhood of $v$. Here $X_v$ means that $X$ acts on the qubit corresponding to $v$, and similarly for $Z_u$.
\end{definition}

Graph states are important since every stabilizer state is, up to local Clifford unitaries, a graph state~\cite{van2004graphical}. As such, to study entanglement in stabilizer states it suffices to understand entanglement in graph states. This 
is convenient because their stabilizers are defined directly from a graph, and many quantum-information quantities become graph-theoretic quantities.

We record a property of graph states. For the sake of convenience, we will denote $P_U:=\prod_{u\in U} P_u$ for a subset $U \subseteq V$ and operators $P_v$ acting on the vertex set $V$. Note that the stabilizers of a graph state can be written as
\begin{align}
S_U = \prod_{u \in U}S_u = \omega \prod_{u\in U} X_u \prod_{v\in \textrm{Odd}(U)}Z_v=\omega X_UZ_{\textrm{Odd}(U)} \label{eq:gen_graph_stab}
\end{align}
for $U\subseteq V$ and $\omega \in \{\pm1\}$, where $\textrm{Odd}(U)$ is the subset of vertices in $V$ that are adjacent to an odd number of vertices in $U$.

We now discuss entanglement across a bipartition. Let $W\subseteq V$ and write $\overline{W} = V\setminus W$. Bipartite entanglement in pure states is captured by the entanglement entropy~\cite{eisert2008area}. For a pure state $\ket{\psi}$, let $\rho_W:=\mathrm{Tr}_{\overline{W}}(\ket{\psi}\bra{\psi})$ be the \emph{marginal} on $W$, i.e., the partial trace over the complement $\overline{W}$ of the projector $\ket{\psi}\bra{\psi}$. The \emph{bipartite entanglement entropy} of $\ket{\psi}$ with respect to the bipartition $W\cup \overline{W}$ is given by the von Neumann entropy of $\rho_W$
\begin{align*}
    -\mathrm{Tr}(\rho_W\log \rho_W).
\end{align*}
Throughout this section we take logarithms in base $2$, so entropy is measured in bits. For graph states the entanglement entropy is also known as the cut-rank.

\begin{definition}[Cut-rank of graph states]\label{def:cutrank_graph_states}
    Let $W\subseteq V$ be a subset of vertices of a graph $G$. The cut-rank $\textrm{cutrk}_G(W)$ is the rank over $\mathbb{F}_2$ of the $W\times \overline{W} $ matrix $M$ with entry $M_{ij}=1$ if $i\in W$ is adjacent to $j \in \overline{W}$ in $G$ and zero otherwise. 
\end{definition}

We now show for completeness that the cut-rank of a subset of vertices of a graph $G$ is equal to the entanglement entropy of the associated bipartition and graph state $\ket{G}$. This is a standard fact (see e.g.~\cite{fattal2004entanglement}). Let $\ket{\psi}$ be a stabilizer state with stabilizer group $\mathcal{S}$. Recall by (\ref{eq:stab_dm}) that $\ket{\psi}\bra{\psi}=\frac{1}{2^n}\sum_{S\in \mathcal{S}}S$. Taking the marginal on $W$ corresponds to taking the partial trace over the complement $\overline{W}$ of $W$. Since taking the partial trace over a Pauli string $\mathrm{Tr}_{\overline{W}}\left(P\right)$ is non-zero if and only if $P$ has support disjoint from $\overline{W}$, the marginal is the subnormalized projector

 \begin{align*}
\rho_W=\textrm{Tr}_{\overline{W}}\left(\ket{\psi}\bra{\psi}\right) =& \frac{1}{2^n}\sum_{S\in \mathcal{S}}\textrm{Tr}_{\overline{W}}\left(S\right)
=\frac{1}{2^{\left|W\right|}}\sum_{S\in \mathcal{S}_W}S
=\frac{1}{2^{\left|W\right|-\textrm{dim}\left(\mathcal{S}_W\right)}}\prod_{i=1}^{\textrm{dim}\left(\mathcal{S}_W\right)}\frac{1+S_i}{2} \ ,
 \end{align*}
where $\mathcal{S}_W = \lbrace S\in \mathcal{S} \mid \textrm{supp}(S)\subseteq W\rbrace $, and $\lbrace{S_i\rbrace}_{i=1}^{\textrm{dim}\left(\mathcal{S}_W\right)}$ is a minimal generating set of $\mathcal{S}_W$. Therefore, the operator $2^{\left|W\right|-\textrm{dim}\left(\mathcal{S}_W\right)}\rho_W$ is a projector onto a subspace of dimension $2^{\left|W\right|-\textrm{dim}\left(\mathcal{S}_W\right)}$, and the entropy of $\rho_W$ is given by
\begin{align}
-\mathrm{Tr}\left(\rho_W \log \rho_W\right) = \left|W\right|-\textrm{dim}\left(\mathcal{S}_W\right)\ ,
\end{align}
since $\rho_W$ has $2^{\left|W\right|-\textrm{dim}\left(\mathcal{S}_W\right)}$ eigenvalues equal to $\frac{1}{2^{{\left|W\right|-\textrm{dim}\left(\mathcal{S}_W\right)}}}$, while the rest are zero.

Given a graph state $\ket{G}$ and a subset $W\subseteq V$, we now show that $\left|W\right| - \dim \left(\mathcal{S}_W\right)$ equals the rank of the matrix $M$ from definition (\ref{def:cutrank_graph_states}). From Eq.~\eqref{eq:gen_graph_stab} and the definition of $\mathcal{S}_W$, we have $S_U\in \mathcal{S}_W$ if and only if $U\cup \textrm{Odd}(U)\subseteq W$. In other words, for every $U\subseteq W$, the stabilizer $S_U \in \mathcal{S}_W$ when $M^T \mathbbm{1}_U =\mathbbm{1}_{\textrm{Odd}(U)\cap \overline{W}}= 0$, where $\mathbbm{1}_U \in \{0,1\}^W$ is the indicator function of $U$. This implies that $\dim\left(\mathcal{S}_W\right)$ equals the nullity of $M$, and thus $\left|W\right| - \dim \left(\mathcal{S}_W\right)$ equals the rank of $M$.

Note that only in the last step did we assume that the underlying state was a graph state. We can thus extend the definition of the cut-rank to arbitrary stabilizer states.

\begin{corollary}[Cut-rank of stabilizer states~\cite{fattal2004entanglement}]\label{corr:cutrk}
    The cut-rank of $W\subseteq V$ for a stabilizer state with stabilizer group $\mathcal{S}$ is equal to $\left|W\right|-\textrm{dim}\left(\mathcal{S}_W\right)$, where $\mathcal{S}_W = \lbrace S\in \mathcal{S} \mid \textrm{supp}(S)\subseteq W\rbrace $.
\end{corollary}

\subsection{Vertex-minors}

Vertex-minor theory has been studied in the math literature, where it has close parallels with structural graph theory/graph minor theory~\cite{mccarty2021local}. Vertex-minor formalizes the effect of single-qubit Clifford operations and single-qubit Pauli measurements on the underlying graph. We first state the definitions when restricted to graphs (such as to make the connection with prior literature more clear), and then generalize this to the setting of arbitrary stabilizer states.

\begin{definition}[Local complementation/Local equivalence]
    For a vertex $v$ in $G$, a \emph{local complementation} of $G$ at a vertex $v$ replaces the subgraph induced by $N_v$ with its complement. That is, for each distinct pair of neighbors $u,w\in N_v$, remove the edge $uw$ if it exists, and otherwise add the edge $uw$ if $u$ is not adjacent to $w$. We say that two graphs $G$ and $H$ are \emph{locally equivalent}, and denote this by $G\sim H$, if $H$ can be obtained from $G$ by a sequence of local complementations.
\end{definition}

Local equivalence is important for quantum information, since local complementations correspond to local Clifford operations on graph states. Let us make this connection more precise. A single-qubit \emph{Clifford operator} is a unitary operator $U$ such that $U\mathcal{P}_1U^{\dagger}=\mathcal{P}_1$, i.e., it is an operator that maps Pauli operators to Pauli operators under conjugacy. A $n$-qubit \emph{local Clifford unitary} is the tensor product $U=U_1\otimes \ldots \otimes U_n$, where each $U_i$ is a single-qubit Clifford operator acting on the qubit $i$. A result in~\cite{van2004graphical} states that two graphs $G$ and $H$ are locally equivalent if and only if there is a local Clifford unitary $U$ such that $\ket{H}=U\ket{G}$.

\medskip

Vertex-minors extend the above notion by incorporating vertex deletions.

\begin{definition}[Vertex-minors]\label{def:vertex_minor}
A graph $H$ is a \emph{vertex-minor} of a graph $G$ if it can be obtained from $G$ by a sequence of local complementations and vertex deletions.
\end{definition}

Equivalently, $H$ is a vertex-minor of $G$ if there exists a graph $G'$ locally equivalent to $G$ such that $H$ is an induced subgraph of $G'$. We write $G \succ H$ in this case. Vertex-minors are central in quantum information since single-qubit Pauli measurements on a graph state $\ket{G}$ produce (up to single-qubit Cliffords) a graph state corresponding to a vertex-minor of $G$. Thus, the vertex-minor relation captures precisely the graph states that can be obtained from $\ket{G}$ by local Clifford operations and (destructive) single-qubit Pauli measurements~\cite{dahlberg2018transforming}. 

\medskip

We remark that the restriction to graph states in the above is unnecessary. Recall that every stabilizer state is equivalent up to single-qubit Cliffords to a graph state, i.e., for every stabilizer state $\ket{\psi}$, there exists a local Clifford unitary $U$ and a graph state $\ket{G}$ such that $\ket{\psi}=U\ket{G}$. Note that two stabilizer states $\ket{\psi_1}$ and $\ket{\psi_2}$ are single-qubit Clifford equivalent if and only if $\ket{\psi_1}$ and $\ket{\psi_2}$ are single-qubit Clifford equivalent to graph states $\ket{G_1}$ and $\ket{G_2}$ such that $G_1$ and $G_2$ are locally equivalent. Similarly, a stabilizer state $\ket{\phi}$ can be reached through single-qubit Clifford unitaries and single-qubit Pauli measurements from a stabilizer state $\ket{\psi}$, if and only if $\ket{\phi}$ and $\ket{\psi}$ are single-qubit Clifford equivalent to graph states $\ket{H}$ and $\ket{G}$ such that $G\succ H$. It is thus possible to extend the notions of local equivalence and vertex-minors to the setting of stabilizer states, for which we will use the same terminology.

\begin{definition}[Local equivalence of stabilizer states]\label{def:local_equivalence}
Two stabilizer states $\ket{\psi_1}$ and $\ket{\psi_2}$ are locally equivalent if they are equivalent up to single-qubit Clifford unitaries.
\end{definition}

\begin{definition}[Vertex-minors of stabilizer states]\label{def:vertex_minor_stab}
A stabilizer state $\ket{\phi}$ is a \emph{vertex-minor} of a stabilizer state $\ket{\psi}$ if $\ket{\phi}$ can be obtained from $\ket{\psi}$ by a sequence of (destructive) single-qubit Pauli measurements and single-qubit Clifford unitaries.
\end{definition}

\section{Single-qubit measurements and rank/nullity functions}\label{sec:single_qubit}

A number of entanglement measures of stabilizer states can be naturally rephrased in terms of the measurement statistics associated with single-qubit Pauli measurements. In this section we lay out the connection between these statistics and locally commutative subgroups, distance, independent sets up to local complementations/the geometric measure of entanglement, and isotropic matroids. Most importantly for this work, we will define the rank and nullity functions attached to stabilizer states.

\subsection{Example: single-qubit Pauli measurements on a small graph state}\label{sec:K3_example}
Here we provide an example of measurement statistics on a small stabilizer state. This will serve primarily as an introduction to such measurement statistics, which we formalize in the following section. 

\medskip

The running example here will be $\ket{K_3}$, the complete graph state on three qubits. This state is equivalent to a GHZ \cite{GreenbergerHorneZeilinger1989} state up to local Clifford unitaries. This state has stabilizer group
\begin{align*}
\mathcal{S}_{\ket{K_3}}&=\langle XZZ, ZXZ, ZZX\rangle \\&
=\{III, XZZ, ZXZ, ZZX, YYI, YIY, IYY, XXX\} ,
\end{align*}
where we ignored phases. We will sometimes refer to phases to aid in exposition in what follows. Our results do not require the tracking of these phases, however. Explicitly incorporating them into the theory we discuss here would be of interest for future work, especially since the understanding of such phases is of importance for contextuality (for example, see~\cite{abramsky2017complete}).

\medskip

Before continuing, we briefly recall measurement statistics of general Pauli operators on stabilizer states, see e.g.~\cite{gottesman1997stabilizer} for more information. Afterwards we will specialize to the case of (sequences of) single-qubit Pauli measurements. Recall that for a Pauli string $P$, the eigenvalues of $P$ are contained in the set \{+1,-1\} and consequently, the outcomes of measuring $P$ are labeled by the two projectors
\begin{align*}
    \Pi_0=\frac{I+P}{2} \qquad \text{and} \qquad \Pi_{1}=\frac{I-P}{2}.
\end{align*}
Therefore, measuring $P$ projects $\ket{\psi}$ into either the $+1$ or $-1$-eigenspace. We label the different outcomes with bits; that is, the bit $0$ is returned if the state is projected into the $+1$-eigenspace, and the bit $1$ is returned if the state is projected into the $-1$-eigenspace.  Moreover, the probability that $\ket{\psi}$ returns a bit $0$ is given by $\textrm{Tr}(\Pi_0\ket{\psi}\bra{\psi})$, while the probability that $\ket{\psi}$ returns a bit $1$ is given by  $\textrm{Tr}(\Pi_1\ket{\psi}\bra{\psi})$. As a consequence, by (\ref{eq:stab_dm}), we obtain that if $P \in \mathcal{S_{\ket{\psi}}}$ or $-P\in \mathcal{S_{\ket{\psi}}}$, then the measurement is deterministic and returns $0$ or $1$, respectively. Otherwise, if $\pm P\notin S_{\ket{\psi}}$, then the measurement returns one of the bits \{0,1\} uniformly at random, both with probability $1/2$.

\medskip

Let us now restrict to single-qubit Pauli measurements. A Pauli string $P=P_1P_2\cdots P_n$ determines a sequence of \emph{single-qubit Pauli measurements}, i.e., we measure the state $\ket{\psi}$ on all single-qubit operators $P_v$ separately (where we interpret $I_v$ as no measurement performed on qubit $v$). Note that this is not the same thing as measuring the corresponding multi-qubit Pauli string $P$ as described in the previous paragraph. Indeed, a multi-qubit Pauli measurement of $P$ produces only one outcome, namely a single bit in $\{0,1\}$. By contrast, measuring the single-qubit Paulis $P_1,\ldots,P_n$ separately produces one outcome for each non-identity factor. 
\medskip

If we measure $XZZ$ on $\ket{K_3}$ as a sequence of single-qubit measurements, we measure $X$ on the first qubit, $Z$ on the second qubit, and $Z$ on the third qubit. Since none of the single-qubit strings are in the stabilizer, the marginal distribution over each individual measurement is uniform. 
However, their product $XZZ$ \emph{is} in the stabilizer, so measuring $XZZ$ as a multi-qubit Pauli measurement would deterministically return an outcome of $0$. As we will see in Lemma \ref{lemma:meas_stats}, this constraint gives us four equally likely bitstrings for the outcomes of the individual measurements: $000, 011, 101, 110$. We will interpret Pauli strings in this work exclusively as sequences of single-qubit Pauli measurements (unless stated otherwise), and thus \emph{measuring in $P$}/\emph{measuring according to $P$} will mean to measure $P$ as a sequence of single-qubit Pauli measurements.

\medskip

For the sake of completeness, consider the case where $Z_1$ has been applied to $\ket{K_3}$. The resultant state will then have $-XZZ$ as a stabilizer. If now $XZZ$ was measured, the resultant outcomes would be one of $001, 010, 100, 111$, uniformly at random. As such, the choice of $\pm XZZ$ encodes one bit of information, which can be extracted through the single-qubit measurement $P=XZZ$ and the parties exchanging classical information.

\medskip

Note that $n$ bits of information can be encoded by choosing the phases for the stabilizers in any generating set. If multi-qubit Pauli measurements are allowed, all of the $n$ bits could be extracted. Due to the locality restriction however, not all bits can be extracted with a single measurement; this obstruction to locally extracting the globally encoded information is an indicator of entanglement, and it will form an important part of this work. 

\medskip

Continuing with our example, consider the single-qubit measurement of $P=YYY$ on $\ket{K_3}$. Since $YYI$, $YIY$ and $IYY$ are all in the stabilizer group $S_{\ket{K_3}}$, their multi-qubit outcomes are deterministic and correspond to the bit $0$. This implies that the single-qubit outcomes for the first, second and third qubit have the constraint that their pairwise sum should give the outcome $0$, which gives us that the outcome in every qubit is the same. Since none of the single-qubit Pauli string $Y$ is in the stabilizer, we obtain that the only outcomes are $000$ and $111$ and they are measured uniformly at random. Note, in this example, that it is not important whether $P=YYY$ itself is a stabilizer or not. Instead, it matters which substrings of $P$ are stabilizers (see Lemma \ref{lemma:substrings_kernel}). Indeed, measuring $P=XYY$ would yield $000$, $100$, $011$ and $111$ uniformly at random, and measuring $P=IYY$ would yield outcomes $00$ or $11$ on the second and third qubit. Measuring $P=YYY$ as a multi-qubit Pauli measurement would have instead yielded a single bit uniformly at random.

\medskip

Let us conclude with an extreme case, namely a single-qubit Pauli measurement such as $P=XYZ$. In this case, no (non-trivial) substring is in the stabilizer and consequently there are no constraints on the outcomes of the single-qubit measurements. Now all possible eight $3$-bit strings appear uniformly at random. At first sight, such measurements seem useless from a quantum perspective: no information is extracted from such a measurement. As we will show, such measurements correspond to \emph{Eulerian vectors}~\cite{bouchet1993compatible, jackson1991supplementary} which have been used in quantum information to prove the complexity of the vertex-minor problem~\cite{dahlberg2022complexity} and the complexity of counting the number of graphs that are locally equivalent to one another~\cite{dahlberg2020counting}. We note that we show this connection for posterity, and will not need it for our work.

\subsection{Single-qubit Pauli measurements and locally commutative subgroups}
To formalize the discussion from the previous section, we will need the following definitions. 

\begin{definition}[Local commutativity]
Two Pauli strings $P$ and $Q$ \emph{locally commute} if they commute entry-wise. Alternatively, two Pauli strings locally commute if, for each entry $v$, they either agree (i.e.~$P_v=Q_v$) or at least one of them is the identity (i.e.~$P_v=I$ or $Q_v=I$). A locally commutative group is a group whose elements pairwise locally commute.
\end{definition}

As an example, $P=IXZ$ and $Q=XIZ$ locally commute, while $P=IXZ$ and $Q=XYZ$ do not. Note also that locally commutativity is not the same as commutativity. For instance, the two strings $P=XZI$ and $Q=ZXX$ commute, while they do not locally commute.

\begin{definition}[Local groups]\label{def:local_group}
Given a Pauli string $P$, the \emph{local group} $\mathcal{C}_P$ is the group generated by the $n$ Pauli strings that agree with $P$ in one entry and are the identity everywhere else. 
\end{definition}

For example, if $P=XYZ$, then the local group is given by $\mathcal{C}_P = \langle XII, IYI, IIZ\rangle$. Note that $\mathcal{C}_P$ is by construction a locally commutative group. Moreover, in the case that the $n$-qubits are indexed by the vertex set $V$, we can write $\mathcal{C}_P=\{P_U\}_{U\subseteq V}$, where $P_U=\prod_{u\in U}P_u$.

\begin{definition}[Intersections between stabilizer groups and local groups]\label{def:intersections}
Let $\mathcal{S}$ be a stabilizer group and $\mathcal{C}_P$ a local group. We denote by $\mathcal{I}(\mathcal{S}, P):=\mathcal{S}\cap\langle \mathcal{C}_P, i\rangle$ the group obtained by the intersection of the stabilizer group and the group containing all the elements of $\mathcal{C}_P$ and its phases.
\end{definition}

Since $\mathcal{C}_P$ is locally commutative, $\mathcal{I}(\mathcal{S}, P)$ is a locally commutative subgroup of $\mathcal{S}$, and in fact every maximal locally commutative subgroup of a stabilizer group arises in this fashion. We remark that even though $\mathcal{I}(S,P)$ is allowed to have more than one phase of a substring of $P$, the fact that $\mathcal{S}$ is a stabilizer group guarantees that only one phase appears and that it can be easily computed as a functional over $\mathbb{F}_2$.

\begin{proposition}\label{prop:functional}
    Let $\mathcal{S}$ be a stabilizer group and $P$ a Pauli string on $n$ qubits. There exists $a \in \mathbb{F}_2^{\textrm{supp}(P)}$ such that
    \begin{align*}
        \mathcal{I}(S,P)\subseteq \{(-1)^{a\cdot \mathbbm{1}_U}P_U:\: U\subseteq \textrm{supp}(P)\}
    \end{align*}
\end{proposition}

\begin{proof}
Fix a set $U\subseteq V$. Suppose that there are two phases of $P_U$ in the stabilizer $\mathcal{S}$, namely $\omega_1P_U$ and $\omega_2P_U$. Since $\mathcal{S}$ is a group, we have that $\omega_1\omega_2I=\omega_1P_U\omega_2P_U \in \mathcal{S}$. However, by definiton we have that $-I \notin \mathcal{S}$. This implies that $\omega_1=\omega_2\in \{-1,+1\}$. Let $\mathcal{U}$ be the family of subsets $U\subseteq \textrm{supp}(P)$ such that a phase of $P_U$ is in $\mathcal{I}(S,P)$ and let $E=\{\mathbbm{1}_U:\: U\in \mathcal{U}\}\subseteq \mathbb{F}_2^{\textrm{supp}(P)}$ be the corresponding set of indicator vectors. The fact that $\mathcal{S}$ is a group implies that $E$ is a subspace of $\mathbb{F}_2^{\textrm{supp}(P)}$. Moreover,
\begin{align*}
    \mathcal{I}(S,P)=\{(-1)^{\alpha(\mathbbm{1}_U)}P_U:\: U\in\mathcal{U}\}.
\end{align*}
for some functional $\alpha: E\rightarrow \mathbb{F}_2$. Therefore, by completing the bases, one can extend $\alpha$ to a functional over $\mathbb{F}_2^{\textrm{supp}(P)}$. Consequently, there exists $a \in \mathbb{F}_2^{\textrm{supp}(P)}$ such that $\alpha(\mathbbm{1}_U)=a\cdot \mathbbm{1}_U$.
\end{proof}

As a consequence of the last proposition, observe that any $\mathcal{I}\left(\mathcal{S}, P\right)$ is naturally isomorphic to a subspace of $\mathbb{F}_2^{\textrm{supp}\left(P\right)}$, where $\textrm{supp}(P)$ is the support of $P$; each Pauli string can be interpreted as a binary string, where the identity takes the role of $0$, and the unique non-identity Pauli (if it exists) in a given entry takes the role of $1$. We will from now on identify the groups $\mathcal{I}\left(\mathcal{S}, P\right)$ with their associated subspaces in $\mathbb{F}_2^{\textrm{supp}(P)}$ to simplify notation. For example, for the graph state $\ket{K_3}$, $\mathcal{I}\left(\mathcal{S}_{\ket{K_3}}, YYY\right) =  \langle YYI, IYY\rangle$ is isomorphic to $\langle 110, 011\rangle\subseteq \mathbb{F}_{2}^V$, while $\mathcal{I}\left(\mathcal{S}_{\ket{K_3}}, YYI\right) = \langle YYI \rangle $, which is isomorphic to $\langle 11\rangle \subseteq \mathbb{F}_2^{\lbrace{1, 2\rbrace}}$.

\medskip

The groups $\mathcal{I}(\mathcal{S}, P)$ determine the measurement statistics when measuring a stabilizer state with stabilizer group $\mathcal{S}$ in $P$, see also~\cite{de2011linearized}.

\begin{lemma}[Statistics of single-qubit Pauli measurements on stabilizer states]\label{lemma:meas_stats}
Let $\ket{\psi}$ be a stabilizer state with stabilizer group $\mathcal{S}\equiv\mathcal{S}_{\ket{\psi}}$ on a set of qubits $V$. Performing single-qubit Pauli measurements according to $P$ on $\ket{\psi}$ yields a measurement outcome uniformly at random from the elements of a coset $a+\mathcal{I}(\mathcal{S}, P)^\perp \in\mathbb{F}_2^{\textrm{supp}\left(P\right)}/\left(\mathcal{I}(\mathcal{S}, P)^\perp\right)$; which coset depends on the phases of the elements in $\mathcal{I}(\mathcal{S}, P)$. 
\end{lemma}
\begin{proof}
For ease of notation, we set $W=\textrm{supp}\left(P\right)$. Recall that a single-qubit Pauli measurement on $v\in W$ is given by the projector $\Pi_{b(v)}=\frac{I+(-1)^{b(v)} P_v}{2}$, where $P_v\in \lbrace{X, Y, Z\rbrace}$ and $b(v)\in \{0,1\}$ is the bit observed after the measurement. Born's rule~\cite{nielsen2010quantum} states that the probability of observing a bitstring $b\in \lbrace{0, 1\rbrace}^{W}$ is given by $\mathrm{Tr}\left[\Pi_b\ket{\psi}\bra{\psi}\right]$, where $\Pi_b=\prod_{v\in W}\Pi_{b(v)}$. A computation shows that
\begin{align*}
    \Pi_b&=\prod_{v\in W} \Pi_{b(v)}= \prod_{\mathclap{v\in W}}~\frac{1+(-1)^{b(v)}P_v}{2}=\frac{1}{2^{w(P)}}\sum_{\mathclap{U\subseteq W}}~\prod_{u\in U}\left(-1\right)^{b(u)}P_u =\frac{1}{2^{w(P)}}~~~\sum_{\mathclap{U\subseteq W}}\left(-1\right)^{b\cdot \mathbbm{1}_U}P_U \ ,
\end{align*}
where we recall that $w(P)$ is the weight of $P$, we set $P_U = \prod_{u\in U}P_u$, and $b\cdot \mathbbm{1}_U=\sum_{v\in W}b(v)\cdot \mathbbm{1}_U(v)$ is the standard inner product over $\mathbb{F}_2$. Applying Born's rule together with (\ref{eq:stab_dm}), we find
\begin{align*}
\mathrm{Tr}\left[\Pi_b\ket{\psi}\bra{\psi}\right] &= \frac{1}{2^{n+w(P)}}\sum_{S\in \mathcal{S}}~\sum_{\mathclap{U\subseteq W}}~\left(-1\right)^{b\cdot \mathbbm{1}_U}\, \mathrm{Tr}\left[S\, P_U\right].
\end{align*}
By using Proposition \ref{prop:functional} together with the fact that $\mathrm{Tr}\left[S\,P_U\right] = \omega 2^n$ if $S=\omega P_U$ for $\omega\in \{\pm1, \pm i\}$, and \(\mathrm{Tr}\left[S \, P_U\right]=0\) otherwise, we have
\begin{align*}
\mathrm{Tr}\left[\Pi_b\ket{\psi}\bra{\psi}\right]= \frac{1}{2^{w(P)}} ~~~\sum_{\mathclap{\mathbbm{1}_U\in \mathcal{I}(\mathcal{S}, P)}} ~~\left(-1\right)^{b \cdot \mathbbm{1}_U+a \cdot \mathbbm{1}_U}
= \frac{1}{2^{w(P)}} ~~~\sum_{\mathclap{\mathbbm{1}_U\in \mathcal{I}(\mathcal{S}, P)}}~~ \left(-1\right)^{c\cdot \mathbbm{1}_U}\ ,
\end{align*}
for some $c\in \mathbb{F}_2^W$.

\medskip

If $c\in \mathcal{I}\left(\mathcal{S}, P\right)^\perp$, then every term in the sum equals $1$, so the sum is $\left|\mathcal{I}\left(\mathcal{S}, P\right)\right|$. If $c\notin \mathcal{I}\left(\mathcal{S}, P\right)^\perp$, choose $\mathbbm{1}_{U_0}\in \mathcal{I}\left(\mathcal{S}, P\right)$ with $c\cdot \mathbbm{1}_{U_0}=1$. A calculation gives
\begin{align*}
    \sum_{\mathbbm{1}_U\in \mathcal{I}(S,P)}(-1)^{c\cdot \mathbbm{1}_U}=\frac{1}{2}\sum_{\mathbbm{1}_U\in \mathcal{I}(S,P)}(-1)^{c\cdot \mathbbm{1}_U}+(-1)^{c\cdot (\mathbbm{1}_U+\mathbbm{1}_{U_0})}=0.
\end{align*}
Therefore
\begin{align}
\mathrm{Tr}\left[\Pi_b\ket{\psi}\bra{\psi}\right] =
\begin{cases}
|\mathcal{I}\left(\mathcal{S}, P\right)|/2^{w(P)}, & b\in a+\mathcal{I}\left(\mathcal{S}, P\right)^\perp,\\
0, & \text{otherwise},
\end{cases}
\label{eq:rank_calc}
\end{align}
which shows that the outcome is uniformly distributed on the coset $a+\mathcal{I}\left(\mathcal{S}, P\right)^\perp$.
\end{proof}

Note that in the above lemma $\mathcal{I}(\mathcal{S},
P)^\perp$ is the dual of $\mathcal{I}(\mathcal{S}, P)$ when interpreted as a subspace of $\mathbb{F}_2^{\textrm{supp}(P)}$. Furthermore, depending on the choice of phases for the stabilizer generators, any coset in $\mathbb{F}_2^{\textrm{supp}(P)}/\left(\mathcal{I}(\mathcal{S}, P)^\perp\right)$ can appear as the support of the associated probability distribution. 

 \medskip

Lemma \ref{lemma:meas_stats} shows that single-qubit Pauli measurements produce an outcome uniformly distributed over a coset of $\mathcal{I}(\mathcal{S},P)^\perp$. Accordingly, $\dim(\mathcal{I}(\mathcal{S},P)^\perp)$ measures the amount of randomness in the outcome: it counts the number of independent random bits produced by the measurement. Dually, $\dim(\mathcal{I}(\mathcal{S},P))$ measures the number of bits learned when measuring $P$,~i.e.~the number of independent phases of the stabilizers. This leads to the following definitions.

\begin{definition}[Rank and nullity of Pauli measurements]\label{def:rank_nullity_meas}
Let $\ket{\psi}$ be a stabilizer state with stabilizer group $\mathcal{S}$ on $n$ qubits, and recall that $w\left(P\right)$ denotes the number of non-identity elements in $P$. The rank and nullity of a single-qubit Pauli measurement $P$ on $\ket{\psi}$ are defined as

\begin{align*}
r(P)\equiv\, & \dim\left(\mathcal{I}\left(\mathcal{S}, P\right)^{\perp}\right) = w(P) - \dim\left(\mathcal{I}\left(\mathcal{S}, P\right)\right) \ , \textrm{ and}\\
\nu(P)\equiv\, 
&w\left(P\right)-r(P)= \dim\left(\mathcal{I}\left(\mathcal{S}, P\right)\right)
\end{align*}
respectively. To simplify notation we do not specify the dependence of $r$ and $\nu$ on $\ket{\psi}$.
\end{definition}

The rank $r(P)$ of a Pauli string $P$ quantifies the randomness of the outcomes obtained by measuring the single-qubit measurements specified by $P$. The possible outcome strings form a coset of size $2^{r(P)}$, on which the outcome distribution is uniform; equivalently, $r(P)$ is the entropy of the underlying distribution. Conversely, the nullity $\nu(P)= w(P) -r(P)$ quantifies how `non-random' the outcomes are. Here $w(P)$ is the largest entropy that in principle could be observed, since each of the $w(P)$ binary measurements could---in principle---yield perfectly random binary outcomes. 

\medskip
Let us consider the $\ket{K_3}$ example from the previous subsection with the above framework in mind. For $P=XZZ$, $\mathcal{I}(\mathcal{S}, P)=\lbrace{III, XZZ\rbrace}\simeq \lbrace{000, 111\rbrace}$, with associated dual $\lbrace{000, 011, 101, 110\rbrace}$. We see that this corresponds to the coset in $\mathbb{F}_2^V/\left(\mathcal{I}(\mathcal{S}, P)^\perp\right)$ containing the identity, since $XZZ$ is a $+1$ stabilizer. When $XZZ$ is a $-1$ stabilizer, the outcomes correspond to the coset $001+\lbrace{000, 011, 101, 110\rbrace} = \lbrace{001, 010, 100, 111\rbrace}$. We see that $r(P) = 2$ and $\nu(P) = 1$. Similarly, $r(YYY) = 1$ and $\nu(YYY) = 2$, with $\mathcal{I}(\mathcal{S}, YYY)^\perp = \lbrace{000, 111\rbrace}$, $r(XYY) = 2$, with $\mathcal{I}(\mathcal{S}, XYY)^\perp = \lbrace{000, 100, 011, 111\rbrace}$, and $r(XYZ) = 3$ and $\nu(XYZ) = 0$ with $\mathcal{I}(\mathcal{S}, XYZ)^\perp = \lbrace{000, 001, 010, 011, 100, 101, 110, 111\rbrace}$. 

\medskip

We finish this subsection by noting that the above rank and nullity functions determine a stabilizer state up to Pauli unitaries.

\begin{proposition}\label{prop:unique_rank}
Let $\ket{\psi}$ and $\ket{\phi}$ be stabilizer states on the same vertex set
$V$. If $r_{\ket{\psi}}(P)=r_{\ket{\phi}}(P)$
for every Pauli string $P$, then $\ket{\psi}$ and
$\ket{\phi}$ are equivalent up to Pauli unitaries.
\end{proposition}

\begin{proof}
Let $\overline{\mathcal S}_{\ket{\psi}}$ be the stabilizer group of $\ket{\psi}$ after forgetting phases. Thus $\overline{\mathcal S}_{\ket{\psi}}$
is a subgroup of $\{I,X,Y,Z\}^V$, and for each
$S\in\overline{\mathcal S}_{\ket{\psi}}$ exactly one of $\pm S$ belongs to the stabilizer group $\mathcal S_{\ket{\psi}}$. We will show that the nullity function determines $\overline{\mathcal S}_{\ket{\psi}}$. This is enough, since $r_{\ket{\psi}}(P)=w(P)-\nu_{\ket{\psi}}(P)$ implies that the nullity function determines the rank function for every Pauli string $P$.

\medskip

Let $P$ be a Pauli string and $U=\operatorname{supp}(P)$ its support. We claim that $\nu_{\ket{\psi}}$ determines the membership of $P$ in $\overline{\mathcal S}_{\ket{\psi}}$. Indeed, since $\nu_{\ket{\psi}}(P)=\dim \mathcal I(\mathcal S_{\ket{\psi}},P)$, we have that $2^{\nu_{\ket{\psi}}(P)}$ counts the number of elements in $\mathcal I(\mathcal S_{\ket{\psi}},P)$. Hence, by an Inclusion-Exclusion argument we have that the number of times that $P$ appears in $\overline{\mathcal S}_{\ket{\psi}}$ is given by
\begin{align*}
    N_{\ket{\psi}}(P)=\sum_{W\subseteq U}
(-1)^{|U|-|W|}
2^{\nu_{\ket{\psi}}(P_W)}.
\end{align*}
That is, $N_{\ket{\psi}}(P)=1$ if $P\in \overline{\mathcal S}_{\ket{\psi}}$ and $N_{\ket{\psi}}(P)=0$ if $P\notin \overline{\mathcal S}_{\ket{\psi}}$. Since the quantity $N_{\ket{\psi}}(P)$ only depends on the value of $\nu_{\ket{\psi}}$ on $P$ and on its substrings, the claim follows.

\medskip

In particular, the above claim implies that if $\ket{\psi}$ and $\ket{\phi}$ are stabilizer states such that $r_{\ket{\psi}}=r_{\ket{\phi}}$, then $P\in \overline{\mathcal S}_{\ket{\psi}}$ if and only if $P\in \overline{\mathcal S}_{\ket{\phi}}$. Hence, it holds that $\overline{\mathcal{S}}\equiv\overline{\mathcal S}_{\ket{\psi}}=\overline{\mathcal S}_{\ket{\phi}}$. At this point only the phases for the elements of $\overline{\mathcal S}_{\ket{\psi}}$ remain to be determined. The $2^n$ possible phases of a generating set can be changed freely by applying Pauli operators, see e.g.~\cite{niekamp2012entropic}. These phases in turn fix the phase of each other group element in $\mathcal{S}_{\ket{\psi}}$. The set of  stabilizer states corresponding to all the $2^n$ options for the phases form an equivalence class under the Pauli equivalence relation. Such equivalence classes are also known as a stabilizer basis~\cite{niekamp2012entropic}.

\end{proof}

We note that the above rank function defines a \emph{multimatroid}~\cite{bouchet1997multimatroids, bouchet1998multimatroids, bouchet2001multimatroids, noble2025primer}. Informally, each local group of a stabilizer state determines a \emph{matroid}, and it is the collection of these local groups that determine important properties of the underlying stabilizer state. Stabilizer states are a particular instance of a multimatroid, and they have been studied from this perspective in the math literature under the name of isotropic systems~\cite{bouchet1987isotropic, bouchet1988graphic, bouchet1989connectivity}, isotropic matroids~\cite{brijder2015isotropic, brijder2015isotropic1} and binary tight $3$-matroids~\cite{brijder2015isotropic1, brijder2018orienting, noble2025primer}.

\subsection{Isotropic matroids}
The discussion from the previous subsections suffices in principle to describe measurement statistics of single-qubit Pauli measurements on stabilizer states. However, the formulation in terms of the intersection group $\mathcal{I}\left(\mathcal{S}, P\right)$ is not particularly convenient for calculations. By phrasing the above discussion in terms of \emph{isotropic matroids} \cite{brijder2015isotropic, brijder2015isotropic1} it is possible to express the rank and nullity of measurements in terms of ranks and nullities of certain binary matrices. This framing is convenient, and we will use it in the proofs of lemmas \ref{lemma:perturbations_affect_rank}, \ref{lemma:perturbations_alpha_kappa} and \ref{lemma:dist_gen_hamming_weights}. We deviate from the definition in~\cite{brijder2015isotropic, brijder2015isotropic1}, and phrase everything in terms of Pauli strings, instead of vertex triples or (sub)transversals.

\medskip

Throughout this subsection all matrices, ranks, and nullities are over $\mathbb F_2$.
Let $G$ be a graph on $n$ vertices, and let $A\equiv A(G)$ be its adjacency matrix over $\mathbb F_2$. For $v\in V$, write $e_v$ for the indicator vector of $v$, and write
$A_v=Ae_v$ for the adjacency vector of $v$, i.e., the indicator vector of the neighborhood $N_v$.

\medskip

We first define a map which associates a binary matrix to each Pauli string. Given a
Pauli string $P$, define $\mathbf M_G(P)\in \mathbb F_2^{V\times V}$ by specifying its
$v$-th column as follows:
\begin{align*}
\mathbf M_G(P)_v
=
\begin{cases}
0, & \text{if } P_v=I,\\
e_v, & \text{if } P_v=Z,\\
A_v, & \text{if } P_v=X,\\
e_v+A_v, & \text{if } P_v=Y.
\end{cases}
\end{align*}
Equivalently, $\mathbf M_G(P)$ is obtained by choosing, for each vertex $v$, one of the
three columns associated with $v$ in the matrix
\begin{align*}
\begin{gathered}
\hspace{7mm} Z \hspace{6mm} X \hspace{10mm} Y\\
IAS(G)=\begin{bmatrix}
I_n \mid A(G) \mid I_n+A(G)
\end{bmatrix}.
\end{gathered}
\end{align*}
The three blocks correspond respectively to the choices $Z$, $X$, and $Y$. That is, the $Z$-column at $v$ is $e_v$, the $X$-column at $v$ is $A_v$, and the
$Y$-column at $v$ is $e_v+A_v$. We remark that the map $\mathbf{M}_G(P)$ ignores the phases, in the sense that $\mathbf{M}_G(P)=\mathbf{M}_G(\omega P)$ for every $\omega$.

\begin{definition}[Isotropic matroid~\cite{brijder2015isotropic1}]\label{def:IAS}
Let $G$ be a graph and let $P$ be a Pauli string on $V$. We define $I(G,P)$ to be
the matrix obtained from $\mathbf M_G(P)$ by deleting the zero columns corresponding to
the vertices $v$ for which $P_v=I$. Equivalently, $I(G,P)$ has one column for each
$v\in\operatorname{supp}(P)$, namely
\[
\begin{cases}
e_v, & \text{if } P_v=Z,\\
A_v, & \text{if } P_v=X,\\
e_v+A_v, & \text{if } P_v=Y.
\end{cases}
\]
\end{definition}

The matrices $I(G, P)$ for $P$ complete are called the transverse matroids of $G$, in the terminology of \cite{brijder2015isotropic1}. As two examples, consider $K_3$ and $C_4$, i.e.~the complete graph on $3$ vertices and the cycle graph on $4$ vertices respectively. Then 

\begin{gather*}
    \hspace*{18.5mm}Z \hspace{12mm} X \hspace{12.5mm} Y\nonumber \\
    IAS(K_3) = \begin{bmatrix}
    ~1 & 0 & 0 ~~|~~ 0 & 1 & 1 ~~|~~ 1  & 1 & 1~\\
    ~0 & 1 & 0 ~~|~~ 1 & 0 & 1 ~~|~~ 1  & 1 & 1~\\
    ~0 & 0 & 1 ~~|~~ 1 & 1 & 0 ~~|~~ 1  & 1 & 1~\\
    \end{bmatrix},\\
    \vspace*{5mm}
    \hspace{18.5mm}Z \hspace{15.7mm} X \hspace{15.5mm} Y\nonumber \\
    IAS(C_4) = \begin{bmatrix}
    ~1 & 0 & 0 & 0 ~~|~~ 0 & 1 & 0 & 1 ~~|~~ 1  & 1 & 0 & 1~\\
    ~0 & 1 & 0 & 0 ~~|~~ 1 & 0 & 1 & 0 ~~|~~ 1  & 1 & 1 & 0~\\
    ~0 & 0 & 1 & 0 ~~|~~ 0 & 1 & 0 & 1 ~~|~~ 0  & 1 & 1 & 1~\\
    ~0 & 0 & 0 & 1 ~~|~~ 1 & 0 & 1 & 0 ~~|~~ 1  & 0 & 1 & 1~\\
    \end{bmatrix}, 
    \end{gather*}

and 
\begin{align*}
I(K_3, XYZ) = \begin{bmatrix}
0 & 1& 0\\
1 & 1& 0\\
1 & 1& 1
\end{bmatrix} \ ,~~~~I(C_4, XIXY) = \begin{bmatrix}
~0 & 0& 1~\\
~1 & 1& 0~\\
~0 & 0& 1~\\
~1 & 1& 1~
\end{bmatrix} \ .
\end{align*}

Let us motivate where the form of the matrix $IAS(G)$ comes from. Recall that the graph state $\ket G$
has vertex-stabilizers $S_v=X_vZ_{N_v}$. Since $S_v\ket G=\ket G$, we have that $X_v\ket G=Z_{N_v}\ket G$.
Thus, when acting on $\ket G$, applying $X_v$ is equivalent to applying $Z$ on the
neighborhood of $v$. Similarly, since $Y_v=X_vZ_v$ up to phases, applying $Y_v$ on $\ket{G}$ is
equivalent, up to phases, to applying $Z$ on $N_v\cup\{v\}$. Therefore, each single-qubit Pauli operator acting on $\ket{G}$ determines a subset of vertices on which a $Z$-operator is applied:
\begin{align*}
Z_v \longleftrightarrow \{v\},\qquad
X_v \longleftrightarrow N_v,\qquad
Y_v \longleftrightarrow N_v\cup\{v\}.
\end{align*}
The matrices $\mathbf{M}_G(P)$ and $IAS(G)$ are simply a way of recording these three possibilities for every vertex.

\medskip

Alternatively, one could see the above action on $\ket{G}$ as a quotient group. To formalize that, we define the column-sum map
\begin{align*}
\partial_G(P)
=\sum_{v\in V} \mathbf M_G(P)_v
\in \mathbb F_2^V.
\end{align*}
Equivalently, $\partial_G(P)=\mathbf M_G(P)\mathbbm 1_V$,
where $\mathbbm 1_V$ is the all-one vector indexed by $V$. The next proposition follows immediately from the definition and we omit the proof.

\begin{proposition}\label{prop:partial_homomorphism}
The map $\partial_G$ is a group homomorphism from Pauli strings modulo phases to
$\mathbb F_2^V$. Equivalently, for any two Pauli strings $P$ and $Q$,
\begin{align*}
    \partial_G(PQ)=\partial_G(P)+\partial_G(Q),
\end{align*}
where the product $PQ$ is taken up to phase.
\end{proposition}

\medskip

Consider the quotient $\mathcal{P}_n/ \langle \mathcal{S}_{\ket{G}},i\rangle$. The quotient group implies an equivalence relation $\sim$ such that $P\sim Q$ if and only if there exists a $S\in \langle\mathcal{S},i\rangle$ such that $P=QS$. In particular, the motivation of $IAS(G)$ made above shows that 
\begin{align*}
    X_v\sim Z_{N_v} \quad \text{and} \quad Y_v\sim Z_{N_v\cup\{v\}}
\end{align*}
Hence, by Proposition \ref{prop:partial_homomorphism}, for every Pauli string $P \in \mathcal{P}_n$ it holds that
\begin{align}\label{eq:PasZoperators}
    P\sim Z_{\partial_G(P)}:=\prod_{v:\partial_G(P)(v)=1}Z_v.
\end{align}
The main observation is that the elements of the stabilizer $\mathcal{S}_{\ket{G}}$ are exactly (up to phases) the Pauli strings whose columns sum to zero.

\begin{lemma}[Kernel of the column-sum map]\label{lemma:kernel_column_sum}
Let $G$ be a graph, $P$ a Pauli string and let $\mathcal S_{\ket{G}}$ be the stabilizer group of the graph
state $\ket G$. Then $P$ is an element of $\mathcal{S}_{\ket{G}}$ up to phases if and only if $\partial_G(P)=0$.
\end{lemma}

\begin{proof}
    Let $P\in \langle\mathcal{S}_{\ket{G}},i\rangle$. Then, it holds that $P=\omega S_U$ for some $U\subseteq V$ and $\omega\in \{\pm 1, \pm i\}$, where 
    \begin{align*}
        S_U:=\prod_{u\in U}S_u=\prod_{u\in U} X_uZ_{N_u}
    \end{align*}
    is the product of the generators of $\mathcal{S}_{\ket{G}}$ for $u\in U$. Since $\partial_G(S_u)=0$, we obtain by Proposition \ref{prop:partial_homomorphism} that
    \begin{align*}
    \partial_G(P)=\sum_{u\in U} \partial_G(S_u)=0.
    \end{align*}
    Now, conversely, suppose that $\partial_G(P)=0$. Then by (\ref{eq:PasZoperators}), we have that $P\sim Z_{\partial_G(P)}=I$
    in $\mathcal{P}_n/\langle\mathcal{S}_{\ket{G}},i\rangle $. This implies that $P\in \langle\mathcal{S}_{\ket{G}},i\rangle$, which concludes the proof.
\end{proof}

We now show how the $I(G, P)$ are related to the nullity function of $\ket{G}$ evaluated at $P$. Recall that if $P$ is a Pauli string and
$U\subseteq \operatorname{supp}(P)$, we write $P_U=\prod_{u\in U}P_u$ for the substring of $P$ indexed by $U$. That is, $P_U$ agrees with $P$ on $U$ and is the
identity outside $U$. Since $I(G,P)$ is obtained by keeping exactly the non-zero columns of $\mathbf M_G(P)$,
we have
\begin{align}\label{eq:partial_IGP}
\partial_G(P_U)=I(G,P)\mathbbm 1_U,
\end{align}
for $U\subseteq\operatorname{supp}(P)$.

\begin{lemma}[Substrings and the kernel of $I(G,P)$]\label{lemma:substrings_kernel}
Let $\ket{G}$ be a graph state with stabilizer group $\mathcal S_G$, and let $P$ be a
Pauli string. Then the subsets $U\subseteq\operatorname{supp}(P)$ for which $P_U$ is a
stabilizer up to phase are exactly the elements of $\ker I(G,P)$. In particular,
\[
\nu(P)=\dim \mathcal I(\mathcal S_G,P)=\operatorname{nullity}(I(G,P)).
\]
\end{lemma}

\begin{proof}
By Lemma~\ref{lemma:kernel_column_sum}, the substring $P_U$ is a stabilizer up to phase
if and only if $\partial_G(P_U)=0$.
By (\ref{eq:partial_IGP}), this is equivalent to
\begin{align*}
    I(G,P)\mathbbm 1_U=0.
\end{align*}
Thus the stabilizer substrings of $P$ are indexed precisely by the kernel of $I(G,P)$.

\medskip

On the other hand, $\mathcal I(\mathcal S_G,P)$ is the subgroup of the stabilizer group
consisting of stabilizers whose non-identity entries are substrings of $P$, with phases
included. Since a stabilizer group contains at most one phase for each Pauli string (i.e.~$P\in \mathcal{S}\implies -P\neq \mathcal{S}$, see Proposition \ref{prop:functional}), the
phases do not change the dimension of this subgroup and the result follows.
\end{proof}

We note that it is in principle possible to extend the above notion of isotropic matroids to general stabilizer states by simply permuting columns. We will not require this for this paper, however.

\section{Entanglement measures}\label{sec:ent_measures}
In the previous section we discussed the general framework for calculating the statistics of single-qubit Pauli measurements on stabilizer states, where in particular the rank and nullity of measurements dictated how much information can be extracted. As we show now, two entanglement measures are naturally captured in this framework: the distance and the locally accessible information. Here we motivate both measures, and show how they can be naturally phrased in terms of the ranks/nullities of Pauli measurements on the stabilizer state.

\subsection{Distance}\label{sec:distance}
The distance of stabilizer states has been studied in many contexts, and is one quantifier of entanglement~\cite{cattaneo2015mindeg, raissi2022general}. 

\begin{definition}
The distance $d$ of a stabilizer state $\ket{\psi}$ is the weight of the smallest non-identity stabilizer in the stabilizer group of $\ket{\psi}$.
\end{definition}

Specialized to graph states, the distance is the smallest degree up to local complementations, plus an additive factor of one~\cite{ cattaneo2015mindeg, javelle2012minimum}. The connection can be seen as follows. For a graph state $\ket{G}$, by~\eqref{eq:gen_graph_stab}, every stabilizer is of the form $S_U=\omega X_{U}Z_{\textrm{Odd}(U)}$ for $\omega\in \{\pm 1, \pm i\}$ and $U\subseteq V$. Hence, the size of the support of $S_U$ is just $w(S_U)=|U\cup \textrm{Odd}(U)|$. In particular, by taking $U=\{v\}$ to be a single vertex, we obtain that $w(S_v)=\deg_G(v)+1$. Since for every graph $H$ locally equivalent to $G$, there exists a local Clifford unitary $U$ such that $\ket{H}=U\ket{G}$ and $U$ preserves the support of Pauli strings, we obtain that the distance $d$ of the stabilizer graph state $\ket{G}$ satisfies
\begin{align*}
d\leq \min_{H\sim G} \delta(H)+1.
\end{align*}
The work of~\cite{ cattaneo2015mindeg, javelle2012minimum} shows that the inequality above is actually an equality.

The next result shows that the distance of a stabilizer state can be expressed also in terms of the nullity function.

\begin{lemma}\label{corollary:dist_nullity}
The distance of a stabilizer state is the smallest weight $w(P)$ of a Pauli string $P$ for which $\nu(P) =1$.
\end{lemma}
\begin{proof}
Let $\mathcal S\equiv\mathcal{S}_{\ket{\psi}}$ be the stabilizer group of $\ket{\psi}$, and let $d=\min\{w(S): S\in \mathcal S,\; S\neq I\}$
be its distance. Recall from Definition~\ref{def:rank_nullity_meas} that $\nu(P)=\dim(\mathcal I(\mathcal S,P))$,
where $\mathcal I(\mathcal S,P)$ consists of the stabilizers which, up to phase, are
substrings of $P$. Our goal is to show that $d=\min\{w(P):\:\nu(P)=1\}$.

We first show that there exists a Pauli string $P$ of weight $d$ such that
$\nu(P)=1$. Let $S\in\mathcal S\setminus\{I\}$ be a stabilizer of weight $d$, and let
$P$ be the phase-free Pauli string underlying $S$. Then $S\in\mathcal I(\mathcal S,P)$,
and therefore $\nu(P)\geq 1$. We claim that in fact $\nu(P)=1$. Suppose, for contradiction, that $\nu(P)\geq 2$.
Then there exists a non-identity substring $P'$ of $P$ which is distinct from $P$ such that $S'=\omega P' \in \mathcal{I}(\mathcal{S},P)$ for some $\omega\in \{\pm1,\pm i\}$. Since $P'$ is a substring of $P$ and is
not equal to $P$ up to phase, its support is a proper subset of $\operatorname{supp}(P)$. Thus $0<w(S')<w(S)=d$,
contradicting the definition of $d$. Hence $\nu(P)=1$, and we obtain that $\min\{w(P):\nu(P)=1\}\leq d$.

Conversely, let $P$ be any Pauli string such that $\nu(P)=1$. Then
$\mathcal I(\mathcal S,P)$ is non-trivial, and hence contains a non-identity stabilizer
$S$ which is, up to phase, a substring of $P$. Therefore $d\leq w(S)\leq w(P)$. Since this holds for every $P$ with $\nu(P)=1$, we obtain that $d\leq \min\{w(P):\nu(P)=1\}$.
Combining the two inequalities gives the desired equality.
\end{proof}

\subsection{Locally accessible information $\alpha_{\rm loc}$}
In the example of the $\ket{K_3}$ graph state in section \ref{sec:K3_example}, it was shown that for the Pauli string $P=YYY$ the nullity satisfies $\nu(P)=2$. Alternatively, performing a measurement $P=YYY$ revealed two bits of information. As can be checked, this is the largest number of bits that can be extracted from single-qubit Pauli measurements. Note that if arbitrary multi-qubit Pauli measurements were allowed three bits of information could have been extracted, since in general $n$ bits of information can be extracted under multi-qubit measurements. 

\begin{definition}\label{corr:alpha_char}
Let $\mathcal{S}$ be the stabilizer group of a state $\ket{\psi}$. The \emph{locally accessible information} is defined as
\begin{align*}    
\alpha_{\rm loc}\left(\ket{\psi}\right) =\max_{P} \mathrm{dim}\left(\mathcal{I}(\mathcal{S},P)\right) = \max_{P}\nu(P), 
\end{align*}
where the maximization in the above can be restricted to complete Pauli strings.
For a stabilizer state with stabilizer group $\mathcal{S}$, the locally accessible information $\alpha_{\rm loc}$ is the largest number of bits that can be extracted using single-qubit Pauli measurements.
\end{definition}

The notion of locally accessible information has been studied before in \cite{markham2007entanglement}. The larger $\alpha_{\rm loc}$ is, the more local the information in the state is. As such, $n-\alpha_{\rm loc}$ can be interpreted as an entanglement measure. We will make this notion more explicit in the next section.

\subsubsection{Connection with other quantifiers of entanglement and graph-theoretic interpretation}\label{sec:connection_with_other_ent_measures}

The locally accessible information $\alpha_{\rm loc}$ governs the largest number of bits that can be extracted using single-qubit Pauli measurements. Alternatively, it was shown in~\cite{brijder2015isotropic1} that $\alpha_{\rm loc}\left(\ket{\psi}\right)$ is the largest independent set in any graph state locally equivalent to $\ket{\psi}$. We provide a different proof of this result using the results developed in the preceding sections.

\begin{proposition}[Locally accessible information and local independence number]\label{prop:alpha_loc_independent}
For a stabilizer state $\ket{\psi}$, the locally accessible information $\alpha_{\rm loc}(\ket{\psi})$ equals the size of the largest independent set maximized over all graph states $\ket{G}$ locally equivalent to $\ket{\psi}$. 
\end{proposition}

\begin{proof}
We first reduce the problem to graph states. Indeed, note that if $\ket{\psi}\sim\ket{\phi}$, i.e., $\ket{\psi}$ is locally equivalent to $\ket{\phi}$, then $\alpha_{\rm loc}(\ket{\psi})=\alpha_{\rm loc}(\ket{\phi})$. This follows from the fact that local Clifford unitaries only permute the single-qubit Pauli operators $\{X,Y,Z\}$ on each qubit, therefore the quantity $\alpha_{\rm loc}$ is invariant under local Clifford equivalence. Let $\alpha_{\rm LC}(G):=\max_{H\sim G}\alpha(H)$ be the largest independent set over all graph states locally equivalent to $\ket{G}$. Hence, we just need to show that $\alpha_{\rm loc}(\ket{G})=\alpha_{\rm LC}(G)$ for every graph state $\ket{G}$.

\medskip

First we show that $\alpha_{\rm loc}(\ket G)\geq \alpha_{\rm LC}(G)$. Note by the above remark, that it suffices to prove that $\alpha_{\rm loc}(\ket H)\geq \alpha(H)$ since $\alpha_{\rm loc}(\ket H)=\alpha_{\rm loc}(\ket G)$ for every $\ket{H}\sim \ket{G}$. Let $W\subseteq V(H)$ be an independent set. We claim that the vertex-stabilizers $\{S_w^H\}_{w\in W}\subseteq \mathcal{S}_{\ket{H}}$ generate a locally commutative subgroup of the stabilizer group $\mathcal{S}_{\ket{H}}$ of $\ket H$. Indeed, if
$w,w'\in W$ are distinct, then $w$ is not adjacent to $w'$. Thus, by~\eqref{eq:gen_graph_stab}, the stabilizer $S_w^H$ has an $X$ on
the qubit $w$, while $S_{w'}^H$ has the identity on the qubit $w$; and similarly with
$w$ and $w'$ reversed. On any vertex outside $W$, the two stabilizers have either
identity or $Z$. Hence, $S_w^H$ and $S_{w'}^H$ commute locally and consequently $\{S_w^H\}_{w\in W}$ is locally commutative.

\medskip

Let $P$ be the Pauli string given by $P_v=X$ if $v\in W$ and $P_v=Z$ for $V\setminus W$. Then, the subgroup generated by $\{S_w^H\}_{w\in W}$ is a subgroup of $\mathcal{I}(\mathcal{S}_{\ket{H}},P)$. Since the stabilizers
$\{S_w^H:w\in W\}$ are algebraically independent, we get
\begin{align*}
\nu_{\ket H}(P)=\dim\mathcal I(\mathcal S_{\ket H},P)\geq |W|.
\end{align*}
Maximizing over all independent sets in $H$ gives us that $\alpha_{\rm loc}(\ket H)\geq \alpha(H)$ and consequently that $\alpha_{\rm loc}(\ket{G})\geq \alpha_{\rm LC}(G)$.

\medskip

We now prove the reverse inequality. Let $P$ be a Pauli string such that $\nu_{\ket G}(P)=\dim\mathcal{I}(\mathcal{S}_{\ket{G}},P)=k$. Write $U=\operatorname{supp}(P)$ and $m=|U|=w(P)$. As discussed earlier, by Proposition \ref{prop:functional}, one can associate $\mathcal{I}(\mathcal{S}_{\ket{G}}, P)$ with a subspace of $\mathbb{F}_2^U$. Let $B$ be the $k\times m$ matrix whose rows form a basis of $\mathcal{I}(\mathcal{S}_{\ket{G}}, P)$. By Lemma~\ref{lemma:meas_stats},
the outcomes obtained by measuring $\ket G$ according to $P$ are uniformly distributed
on a coset of $\mathcal I(\mathcal S_{\ket G},P)^\perp
\subseteq \mathbb F_2^U$. Therefore, this coset has dimension $m-k$ and can be written as
\begin{align*}
    \{x\in\mathbb F_2^U:Bx=a\}
\end{align*}
for some $a\in \mathbb{F}_2^k$ (the choice of $a$ depends on the phases of the stabilizers). Since
$\operatorname{rank}(B)=k$, there is a set $W\subseteq U$ of size $k$ such that the $k\times k$
submatrix $B_W$ (obtained by restricting to the columns indexed by $W$) is invertible. Then, for every choice of
$x_{U\setminus W}\in\mathbb F_2^{U\setminus W}$, there is a unique choice of $x_W\in\mathbb F_2^W$ satisfying
\begin{align}\label{eq:measuring_deterministic}
B_Wx_W=a+B_{U\setminus W}x_{U\setminus W}.
\end{align}

\medskip

We now proceed by performing single-qubit Pauli measurements on $\ket{G}$ using the Pauli string $P$ on the qubits $U\setminus W$. Let $\ket{\phi}$ be the state obtained after measuring the qubits in $U\setminus W$. Equation~\eqref{eq:measuring_deterministic} guarantees that $x_W$ is completely determined given the outcome $x_{U\setminus W} \in \mathbb{F}_2^{U\setminus W}$ obtained in the measurement. Therefore, the measurement on $\ket{\phi}$ performed in the qubits of $W$ are deterministic. In the language of Lemma~\ref{lemma:meas_stats} this is the same as saying that $\mathcal{I}(\mathcal{S}_{\ket{\phi}},P_W)=\mathbb{F}_2^W$. In particular, this implies that $\pm P_w \in \mathcal{S}_{\ket{\phi}}$ for every $w\in W$.

\medskip

By Definitions~\ref{def:vertex_minor}--\ref{def:vertex_minor_stab}, the
post-measurement state $\ket{\phi}$ is locally equivalent to a graph state $\ket{F}$,
where $F$ is a vertex-minor of $G$. Since local Clifford unitaries preserve the
support of Pauli strings, the graph state $\ket F$ also has a single-qubit
stabilizer supported on each $w\in W$. We claim that this implies that $W$ is an independent set in $F$. Indeed, by \eqref{eq:gen_graph_stab}, every
stabilizer of $\ket F$ has the form $S^F_U=\omega X_UZ_{\operatorname{Odd}_F(U)}$ for some $\omega\in \{\pm1, \pm i\}$. If such a stabilizer is supported only on $w \in W$, then necessarily $U=\{w\}$, and
so the support of $S^F_U$ is $\{w\}\cup N_F(w)$. Hence $N_F(w)=\emptyset$. Thus the vertices in $W$ are isolated in $F$, and in particular form an
independent set. 

Since $F$ is a vertex-minor of $G$, there is a graph
$H$ locally equivalent to $G$ such that $F$ is an induced subgraph of $H$. Therefore $H$ contains an independent set of size $|W|=k$, and hence
\begin{align*}
\alpha_{\rm loc}(\ket{G})=\nu_{\ket{G}}(P)=k\leq \max_{H\sim G}\alpha(H).
\end{align*}
This concludes the proof of the proposition.
\end{proof}

Recall from Definition \ref{def:vertex_minor_stab} that $\ket{\phi}$ is a vertex-minor of $\ket{\psi}$ if and only if $\ket{\psi}$ and $\ket{\phi}$ are locally equivalent to graph states $\ket{G}$ and $\ket{H}$ respectively, such that $H$ is an induced subgraph of $G$. A consequence of Proposition~\ref{prop:alpha_loc_independent} is that the largest unentangled vertex-minor of $\ket{\psi}$ is the largest independent set of a graph state $\ket{G}$ locally equivalent to $\ket{\psi}$, which is of size $\alpha_{\rm loc}(\ket{\psi})$. Conversely, $n-\alpha_{\rm loc}\left(\ket{\psi}\right)$ is equal to the smallest number of measurements needed to disentangle the state, a quantity known as the Pauli persistency~\cite{briegel2001persistent, hein2006entanglement, tzitrin2018local, cabello2009entanglement}. Furthermore, independent sets up to local complementations have also appeared in the context of state distribution~\cite{prielinger2025piecemaker, goodenough2026exact}. More specifically, the simplicial complex consisting of subsets $X\subseteq V$ that can be made to be independent up to local complementations dictates how difficult it is to distribute the state.

\medskip

We finally note that $\alpha_{\rm loc}$ bounds the geometric measure of entanglement~\cite{weinbrenner2025quantifying, weinbrenner2026complete}, which is defined as 
\begin{align*}
E(\ket{\psi})\equiv -\log\Big(\max_{\ket{\phi}\in \mathrm{PROD}} |\braket{\phi}{\psi}|^2\Big).
\end{align*}
Here $\mathrm{PROD}\subseteq \mathbb{C}^{2^n}$ is the subset of pure product states, i.e.~every state of the form $\ket{\psi} = \ket{\psi_1}\otimes \cdots \otimes \ket{\psi_n}$. 
This notion captures `how close' a state is to a product state. To see the relation between the geometric measure of entanglement and the locally accessible information, start by fixing a stabilizer state $\ket{\psi}$ with stabilizer group $\mathcal{S}$. Let $P$ be a Pauli string of full weight, i.e., $w(P)=n$, and write $k=\nu(P)=\dim \mathcal I(\mathcal S,P)$. For each
$b\in \mathbb F_2^n$, consider the projector
\begin{align*}
\Pi_b^P=\prod_{v\in V}\frac{I+(-1)^{b(v)}P_v}{2}.
\end{align*}
Let $\mathcal{C}_{P,b}$ be the locally commutative group generated by $\{(-1)^{b(v)}P_v\}_{v\in V}$. Since $w(P)=n$, we have $P_v\neq I$ for every $v\in V$. Hence, as a consequence of the discussion from Section~\ref{sec:formalism}, there is a unique (up to phases) stabilizer state $\ket{\phi_{P,b}}$ whose stabilizer group is $\mathcal{C}_{P,b}$. In particular, by~\eqref{eq:stab_dm}, we have
\begin{align*}
\Pi_b^P=\ket{\phi_{P,b}}\bra{\phi_{P,b}}.
\end{align*}
Moreover, since each factor of the form $\frac{I+(-1)^{b(v)}P_v}{2}$ acts only on the qubit indexed by $v$, the stabilizer state $\ket{\phi_{P,b}}$ is a product state.

\medskip

By Lemma~\ref{lemma:meas_stats}, the outcomes of the single-qubit Pauli measurement specified by $P$ on $\ket{\psi}$ are uniformly distributed on a coset $a+\mathcal I(\mathcal S,P)^\perp\subseteq \mathbb F_2^n$ for some $a\in \mathbb{F}_2^{n}$. In particular, for any $b\in a+\mathcal I(\mathcal S,P)^\perp$, we have
\begin{align*}
\operatorname{Tr}\left(\Pi_b^P\ket{\psi}\bra{\psi}\right)
=
\frac{|\mathcal I(\mathcal S,P)|}{2^n}
=
\frac{2^k}{2^n}
=
2^{-(n-k)}.
\end{align*}
Therefore, since $\ket{\phi_{P,b}}\in \textrm{PROD}$, it holds that
\begin{align*}
\max_{\ket{\phi}\in \mathrm{PROD}}
|\braket{\phi}{\psi}|^2
\geq
|\braket{\phi_{P,b}}{\psi}|^2
= \operatorname{Tr}\left(\Pi_b^P\ket{\psi}\bra{\psi}\right)=
2^{-(n-k)}.
\end{align*}
Taking logarithms and maximizing over all complete Pauli strings $P$ gives
\begin{align*}
E(\ket{\psi})\leq n-\alpha_{\rm loc}(\ket{\psi}).
\end{align*}

\subsection{Generalized Hamming weights}\label{def:hamming_weights}
For the proof of lemma \ref{lemma:dist_gen_hamming_weights}, it will be convenient to extend the notion of distance defined in Subsection~\ref{sec:distance}. This lemma in turn is key in our proof that the distance can at most grow logarithmically when forbidding a vertex-minor. We will use generalized Hamming weights, a well-known concept from classical coding theory~\cite{wei1991generalized}.

\begin{definition}
Let $\ket{\psi}$ be a stabilizer state, and let $\ell\in \mathbb{N}$ be such that $1\leq \ell \leq \alpha_{\rm loc}(\ket{\psi})$. The generalized Hamming weights $d_\ell$ of $\ket{\psi}$ are the minimum weights of the Pauli strings with nullity $\ell$, i.e.~

\begin{align}
d_\ell \equiv \min_{\mathclap{\substack{P \textrm{~s.t. }\\\nu(P)=\ell}}} ~w(P)\ .
\end{align}
\end{definition}

Clearly the first generalized Hamming weight $d_1$ equals the distance $d$, and the largest $\ell$ is given by $\alpha_{\rm loc}(\ket{\psi})$, since by definition there exist no Pauli strings with larger nullity than $\alpha_{\rm loc}$. The generalized Hamming weights thus contain more information than the distance and locally accessible information, and might thus be of independent interest to the quantum community

\section{Single-qubit Pauli measurements on circle stabilizer states}\label{sec:single_qubit_on_circle}
In section \ref{sec:ent_measures} we showed how the ranks and nullities of all sequences of single-qubit Pauli measurements revealed important information about the underlying state. We now show that so-called circle stabilizer states are precisely the states whose rank functions can, by definition, be represented by a rank function associated with $4$-regular multigraphs. This in turn allows us to express (bounds on) the distance $d$ and the locally accessible information $\alpha_{\rm loc}$ of circle stabilizer states in terms of graph-theoretic properties of $4$-regular multigraphs. We then show that the structure of $4$-regular multigraphs necessarily limits the behavior of the associated rank function, which in turn imposes strong constraints on the asymptotic behavior of the distance and locally accessible information.

\subsection{Rank functions of $4$-regular multigraphs/circle graph states}

A \emph{half-edge} is one of the two ends of an edge. If an edge joins two distinct vertices $u$ and $v$, then one of its half-edges is incident with $u$ and the other is incident with $v$. If the edge is a loop at a vertex $v$, then both of its
half-edges are incident with $v$. A \emph{$4$-regular multigraph} is a multigraph, possibly containing parallel
edges and loops, in which exactly four half-edges are incident with every vertex. Equivalently, every vertex has degree $4$, where a loop contributes $2$ to the degree of its incident vertex.

\medskip

Note that for each vertex $v\in V$ of a $4$-regular multigraph $F$, there are three possible pairings of the four incident half-edges to $v$. We label the three possible choices arbitrarily by $X_v, Y_v, Z_v$, and will say that $F$ is \emph{coded} if such a choice has been made for each vertex. We follow the definition of detachments of coded $4$-regular multigraphs given in the work of Bouchet~\cite{bouchet2001multimatroids}.

\begin{definition}[Detachments of coded $4$-regular multigraphs]\label{def:detachment}
Let $F$ be a coded $4$-regular multigraph on the vertex set $V$ and $P$ a Pauli string on the same set of vertices with support $U=\operatorname{supp}(P)\subseteq V$. The detachment of $F$ with respect to a splitter $P$ is the graph $F||P$ obtained after replacing each vertex $v\in U=\operatorname{supp}(P)$ by two vertices $v'$ and $v''$, such that $v'$ is incident to the two half-edges of one pairing identified by $P_v$, and $v''$ is incident to the two remaining half-edges of the pairing.
\end{definition}

We depict the three possible detachments at a vertex in Fig.~\ref{fig:detachments}.

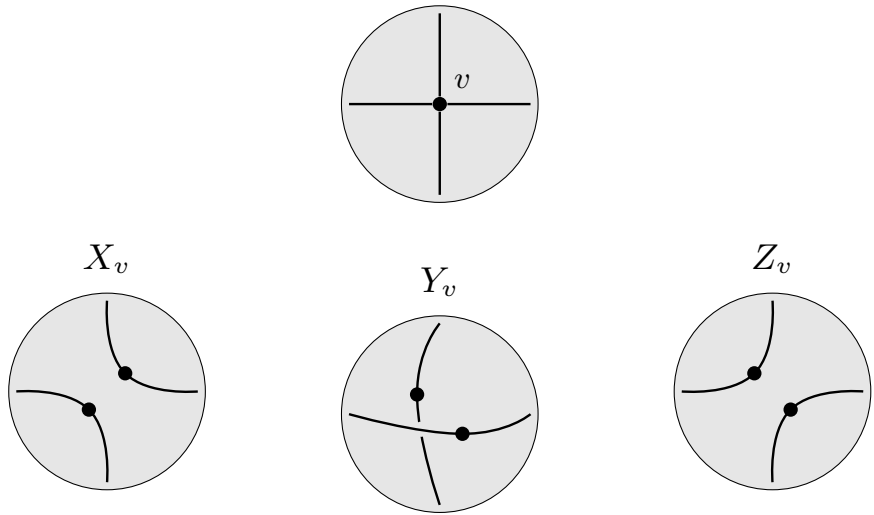
\begin{figure}[h!]
    \centering
    \begin{tikzpicture}[
    scale=1,
    every node/.style={font=\small},
    vtx/.style={circle, fill=black, inner sep=1.9pt},
    hed/.style={line width=0.9pt}
]

\begin{scope}[xshift=4.4cm, yshift=2.8cm]
\filldraw[fill=gray!20, draw=black] (0,0) circle (1.3cm);
    \coordinate (N) at (0,1.2);
    \coordinate (E) at (1.2,0);
    \coordinate (S) at (0,-1.2);
    \coordinate (W) at (-1.2,0);

    \node[vtx] (v) at (0,0) {};
    \draw[hed] (v)--(N);
    \draw[hed] (v)--(E);
    \draw[hed] (v)--(S);
    \draw[hed] (v)--(W);
    \node[scale=1.3] at (0.3,0.3) {$v$};
\end{scope}

\begin{scope}[xshift=0cm, yshift=-1cm]
\filldraw[fill=gray!20, draw=black] (0,0) circle (1.3cm);
    \coordinate (N) at (0,1.2);
    \coordinate (E) at (1.2,0);
    \coordinate (S) at (0,-1.2);
    \coordinate (W) at (-1.2,0);

    \node[vtx] (x1) at (0.24,0.24) {};
    \node[vtx] (x2) at (-0.24,-0.24) {};

    \draw[hed] plot[smooth, tension=0.9] coordinates {(N) (x1) (E)};
    \draw[hed] plot[smooth, tension=0.9] coordinates {(S) (x2) (W)};

    \node[scale=1.5] at (0,1.75) {$X_v$};
\end{scope}

\begin{scope}[xshift=4.4cm, yshift=-1.3cm]
\filldraw[fill=gray!20, draw=black] (0,0) circle (1.3cm);
    \coordinate (N) at (0,1.2);
    \coordinate (E) at (1.2,0);
    \coordinate (S) at (0,-1.2);
    \coordinate (W) at (-1.2,0);

    \node[vtx] (y1) at (-0.30,0.26) {};

    \node[vtx] (y2) at (0.30,-0.26) {};

    \draw[hed] plot[smooth, tension=0.9] coordinates {(N) (y1) (S)};
    \filldraw[fill=gray!20, draw=none] (-0.25,-0.2) circle (0.1cm);
    
    \draw[hed] plot[smooth, tension=0.9] coordinates {(W) (y2) (E)};

    \node[scale=1.5] at (0,1.75) {$Y_v$};
\end{scope}

\begin{scope}[xshift=8.8cm, yshift=-1cm]
\filldraw[fill=gray!20, draw=black] (0,0) circle (1.3cm);
    \coordinate (N) at (0,1.2);
    \coordinate (E) at (1.2,0);
    \coordinate (S) at (0,-1.2);
    \coordinate (W) at (-1.2,0);

    \node[vtx] (z1) at (-0.24,0.24) {};
    \node[vtx] (z2) at (0.24,-0.24) {};

    \draw[hed] plot[smooth, tension=0.9] coordinates {(N) (z1) (W)};
    \draw[hed] plot[smooth, tension=0.9] coordinates {(E) (z2) (S)};

    \node[scale=1.5] at (0,1.75) {$Z_v$};
\end{scope}

\end{tikzpicture}
\caption{The three detachments at a vertex $v$. The under/overcrossing in the $Y_v$ detachment is only present for visual clarity. We have suppressed the labels of $v'$ and $v''$.
}
\label{fig:detachments}
\end{figure}

One can now associate a rank function to splitters on coded $4$-regular multigraphs.

\begin{definition}[Rank and nullities of splitters on coded $4$-regular multigraphs]\label{def:rank_null_4reg}
The rank function of a splitter $P$ associated with a coded $4$-regular multigraph $F$ on vertex set $V$ is defined as
\begin{align}
r_F(P) = w(P) - c(F||P)+c(F) \ ,
\end{align}
where $c(G)$ is the number of connected components of the multigraph $G$. Similarly, the nullity of $P$ associated with $F$ is defined as
\begin{align}\label{eq:nullity_F}
    \nu_F(P) \equiv w(P)-r_F(P)= c(F||P)-c(F)\ .
\end{align}
\end{definition}

We remark that if $F$ is not connected, then there exists another coded $4$-regular multigraph $F'$ on the same vertex set that is connected, such that $r_F = r_{F'}$ (see Proposition 3.7 of \cite{bouchet2001multimatroids}). As such, we will assume that $F$ is connected, i.e., $c(F)=1$, in the remainder of this work. Furthermore, we will also implicitly assume that every $4$-regular multigraph $F$ has already a fixed labeling of each pairing for all $v\in V$, i.e., that it is coded.

\medskip

Surprisingly, Bouchet showed that such rank functions of coded $4$-regular multigraphs are in one-to-one correspondence with the rank function of stabilizer states locally equivalent to so-called \emph{circle} graph states~\cite{bouchet2001multimatroids, bouchet1988graphic} (see Proposition 33 from~\cite{brijder2022characterization}). To be more precise, we define a circle graph as follows. A \emph{chord diagram} is a finite collection of chords of a circle, where a chord is a line segment joining two points of the circle. Given a chord diagram $\mathcal D$, its \emph{intersection graph} is the graph with
vertex set $\mathcal D$ in which two distinct chords are adjacent if and only if they cross an odd number of times in the interior of the circle. A graph $G$ is a \emph{circle graph} if it is the intersection graph of some chord diagram (see Figure~\ref{fig:chord_diagram}). Alternatively, a result of~\cite{bouchet1994circle} shows that circle graphs are precisely the graphs that do not contain any of the graphs in Figure~\ref{fig:forbidden_VM_circle} as a vertex-minor.

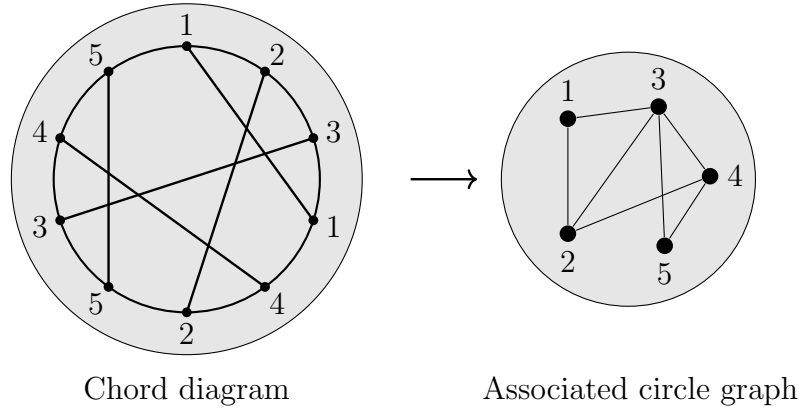
\begin{figure}[h!]
    \centering
    \begin{tikzpicture}[
    scale=0.8,
    vertex/.style={circle, fill=black, inner sep=2.2pt},
    endpoint/.style={circle, fill=black, inner sep=1.3pt},
    every node/.style={font=\large}
]

\begin{scope}[xshift=0cm]
    \def\r{2.2}
\filldraw[fill=gray!20, draw=black] (0,0) circle (2.9cm);
    \draw[line width=0.8pt] (0,0) circle (\r);

    \coordinate (p1)  at ( 90:\r);
    \coordinate (p2)  at ( 54:\r);
    \coordinate (p3)  at ( 18:\r);
    \coordinate (p4)  at (-18:\r);
    \coordinate (p5)  at (-54:\r);
    \coordinate (p6)  at (-90:\r);
    \coordinate (p7)  at (-126:\r);
    \coordinate (p8)  at (-162:\r);
    \coordinate (p9)  at ( 162:\r);
    \coordinate (p10) at ( 126:\r);

    \foreach \p in {p1,p2,p3,p4,p5,p6,p7,p8,p9,p10}
        \node[endpoint] at (\p) {};

    \draw[line width=0.9pt] (p1) -- (p4);
    \draw[line width=0.9pt] (p2) -- (p6);
    \draw[line width=0.9pt] (p3) -- (p8);
    \draw[line width=0.9pt] (p5) -- (p9);
    \draw[line width=0.9pt] (p7) -- (p10);

    \node at ( 90:{\r+0.35}) {$1$};
    \node at (-18:{\r+0.35}) {$1$};

    \node at ( 54:{\r+0.35}) {$2$};
    \node at (-90:{\r+0.35}) {$2$};

    \node at ( 18:{\r+0.35}) {$3$};
    \node at (-162:{\r+0.35}) {$3$};

    \node at (-54:{\r+0.35}) {$4$};
    \node at (162:{\r+0.35}) {$4$};

    \node at (-126:{\r+0.35}) {$5$};
    \node at (126:{\r+0.35}) {$5$};

    \node at (0,-3.5) {Chord diagram};
\end{scope}

\draw[->, line width=0.9pt] (3.7,0) -- (4.8,0);

\begin{scope}[xshift=7.3cm]
\filldraw[fill=gray!20, draw=black] (0,0) circle (2.1cm);

    \node[vertex,label=above:$1$] (v1) at (-1.0, 1.0) {};
    \node[vertex,label=below:$2$] (v2) at (-1.0,-0.9) {};
    \node[vertex,label=above:$3$] (v3) at ( 0.5, 1.2) {};
    \node[vertex,label=right:$4$] (v4) at ( 1.35,0.05) {};
    \node[vertex,label=below:$5$] (v5) at ( 0.6,-1.1) {};

    \draw (v1)--(v2);
    \draw (v1)--(v3);
    \draw (v2)--(v3);
    \draw (v2)--(v4);
    \draw (v3)--(v4);
    \draw (v3)--(v5);
    \draw (v4)--(v5);

    \node at (0.2,-3.5) {Associated circle graph};
\end{scope}

\end{tikzpicture}
\caption{A chord diagram and its associated circle graph.
}
\label{fig:chord_diagram}
\end{figure}

\begin{figure}[h!]
    \centering
    \begin{tikzpicture}[
    scale=1.15,
    vertex/.style={circle, fill=black, inner sep=2pt},
    labelstyle/.style={font=\small}
]

\begin{scope}[xshift=0cm]
\filldraw[fill=gray!20, draw=black] (0,0) circle (1.5cm);
    \node[vertex] (w5c) at (0,0) {};
    \node[vertex] (w51) at (90:1.25) {};
    \node[vertex] (w52) at (162:1.25) {};
    \node[vertex] (w53) at (234:1.25) {};
    \node[vertex] (w54) at (306:1.25) {};
    \node[vertex] (w55) at (18:1.25) {};

    \draw (w51)--(w52)--(w53)--(w54)--(w55)--(w51);
    \draw (w5c)--(w51);
    \draw (w5c)--(w52);
    \draw (w5c)--(w53);
    \draw (w5c)--(w54);
    \draw (w5c)--(w55);

\end{scope}

\begin{scope}[xshift=4.2cm]
\filldraw[fill=gray!20, draw=black] (0,0) circle (1.5cm);
    \node[vertex] (bwc) at (0,0) {};
    \node[vertex] (bw1) at (90:1.25) {};
    \node[vertex] (bw2) at (150:1.25) {};
    \node[vertex] (bw3) at (30:1.25) {};
    \node[vertex] (bw4) at (-30:1.25) {};
    \node[vertex] (bw5) at (-90:1.25) {};
    \node[vertex] (bw6) at (-150:1.25) {};

    \draw (bw2)--(bw1)--(bw3)--(bw4)--(bw5)--(bw6)--(bw2);
    \draw (bwc)--(bw2);
    \draw (bwc)--(bw3);
    \draw (bwc)--(bw5);

\end{scope}

\begin{scope}[xshift=8.7cm]
\filldraw[fill=gray!20, draw=black] (0,0) circle (1.5cm);
    \node[vertex] (w7c) at (0,0) {};
    \node[vertex] (w71) at (90:1.25) {};
    \node[vertex] (w72) at (141.4286:1.25) {};
    \node[vertex] (w73) at (192.8571:1.25) {};
    \node[vertex] (w74) at (244.2857:1.25) {};
    \node[vertex] (w75) at (295.7143:1.25) {};
    \node[vertex] (w76) at (347.1429:1.25) {};
    \node[vertex] (w77) at (38.5714:1.25) {};

    \draw (w71)--(w72)--(w73)--(w74)--(w75)--(w76)--(w77)--(w71);
    \draw (w7c)--(w71);
    \draw (w7c)--(w72);
    \draw (w7c)--(w73);
    \draw (w7c)--(w74);
    \draw (w7c)--(w75);
    \draw (w7c)--(w76);
    \draw (w7c)--(w77);

\end{scope}

\end{tikzpicture}
\caption{The three forbidden vertex-minors for circle stabilizer states.
}
\label{fig:forbidden_VM_circle}
\end{figure}
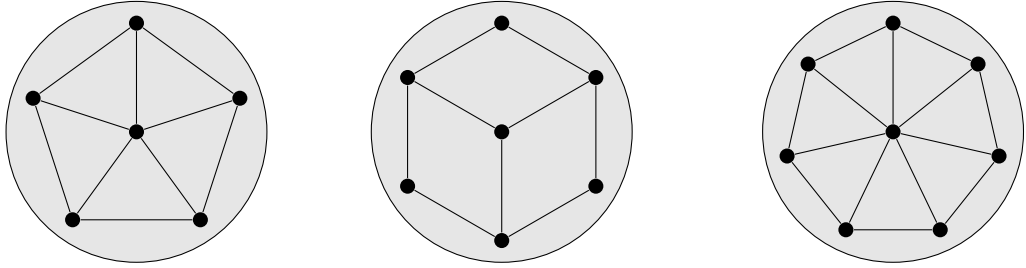

Circle graphs are related to $4$-regular multigraphs in the following way. Given a chord
diagram with chords labelled by a set $\mathcal{D}$ of size $n$, read the labels of the chord endpoints in cyclic order around the circle. This gives a cyclic word $w_1w_2\cdots w_{2n}$ in which every element of $\mathcal{D}$ appears exactly twice. Now construct a multigraph $F$ with vertex set $\mathcal{D}$ by adding an
edge joining $w_i$ to $w_{i+1}$, for each $i \in [2n]$, with indices taken cyclically. Then $F$ is $4$-regular, since every vertex appears twice in the word, and the cyclic word
defines an Eulerian tour\footnote{Or Eulerian system of circuits more generally, i.e.~an Eulerian tour for each connected component~\cite{brijder2022characterization}.} $C$ of $F$. Conversely,
given an Eulerian tour $C$ of a connected $4$-regular multigraph $F$, the cyclic order in which $C$ visits the vertices gives a cyclic word in which every vertex appears twice; drawing a chord between the two occurrences of each vertex produces a
chord diagram. Therefore, there is a correspondence between chord diagrams $\mathcal{D}$ and $4$-regular graphs $F$ with an Eulerian tour $C$.

\medskip

We are now able to properly define circle stabilizer states.

\begin{definition}[Circle stabilizer states]
    A stabilizer state $\ket{\psi}$ is a \emph{circle stabilizer state} if it is locally equivalent to a graph state $\ket{G}$ where $G$ is a circle graph.
\end{definition}

The main result of this subsection that will be useful for us in the subsequent subsections, is the following translation of a characterization of circle stabilizer states due to Bouchet~\cite{bouchet2001multimatroids, bouchet1988graphic}.

\begin{proposition}\label{prop:circle_characterization}
Let $\ket{\psi}$ be a stabilizer state, and let $r_{\ket{\psi}}$ be its rank function. Then $\ket{\psi}$ is a \emph{circle stabilizer state} if and only if there exists a coded $4$-regular multigraph $F$ whose rank function $r_F$ satisfies $r_{\ket{\psi}} = r_F$.
\end{proposition}

Note that the above can be used as an alternative definition of circle stabilizer states, which does not rely on a graph-theoretic characterization.
We are not going to provide a formal proof of Proposition~\ref{prop:circle_characterization} here. However, in the following, we attempt to give a brief sketch by translating the results in the literature in our notation. Suppose that $\ket{\psi}$ is locally equivalent to $\ket{G}$, where $G$ is a circle graph. As discussed above, there exists a $4$-regular multigraph $F$ and an Eulerian tour $C$ corresponding to $G$. A result of Bouchet states that $r_{\ket{G}}=r_{F}$ for some choice of coding of $F$  (see Proposition 33 in~\cite{brijder2022characterization}). We remark that the result in~\cite{brijder2022characterization} is written in the language of multimatroids. In particular, the isotropic $3$-matroid $\mathcal{Z}_3(G)$ has the rank defined in Definition~\ref{def:rank_nullity_meas} (and more specifically its nullity from Lemma~\ref{lemma:substrings_kernel}), while the Eulerian $3$-matroid $Q(F)$ has the rank stated in Definition~\ref{def:rank_null_4reg}. We note that~\cite{brijder2022characterization} restricts to full-weight Pauli strings, but this is sufficient to reconstruct the ranks of all Pauli strings. The converse proceeds as follows. Suppose that there exists a coded $4$-regular multigraph $F$ such that $r_{\ket{\psi}}=r_F$. Then, by the correspondence above and the result in~\cite{brijder2022characterization}, there exists a circle graph $G$ and a recoding $F'$ of $F$ such that $r_{\ket{G}}=r_{F'}$. Since a coding is just a relabeling of the Pauli's, there exists a stabilizer state $\ket{\phi}$ locally equivalent to $\ket{G}$ whose rank function $r_{\ket{\phi}}$ is isomorphic to $r_F$. Here isomorphic means that the functions are equivalent up to a local relabeling of the Pauli's, which can always be achieved through single-qubit Clifford unitaries.
Hence, by Proposition~\ref{prop:unique_rank}, it follows that $\ket{\psi}$ is locally equivalent to $\ket{G}$.

\medskip

We give an example of a $4$-regular multigraph associated to the $\ket{K_3}$ state in Figure~\ref{fig:detachment_example_appendix}, where we also show the detachments according to the splitters $YYY, XXX, XYZ$. Note that these have $3$, $2$ and $1$ connected components, respectively; this matches the expected nullities of $\nu(YYY)=2$, $\nu(XXX)=1$ and $\nu(XYZ) = 0$ found in Section \ref{sec:single_qubit}. Note that this uses a different coding from the example in Figure~\ref{fig:detachment_example}.

\begin{figure}[h!]
    \centering

\begin{tikzpicture}[
    scale=0.8,
    every node/.style={font=\small},
    vtx/.style={circle, fill=black, inner sep=1.9pt},
    hed/.style={line width=0.9pt, line cap=round, line join=round},
    arr/.style={-{Latex[length=2.7mm,width=2mm]}, line width=0.9pt}
]


\begin{scope}[xshift=4.4cm, yshift=3.2cm]
    \filldraw[fill=gray!20, draw=black] (0,0) circle (1.55cm);

    \coordinate (T) at (0,0.95);
    \coordinate (L) at (-0.95,-0.62);
    \coordinate (R) at (0.95,-0.62);

    \draw[hed] (T) .. controls (-0.70,0.75) and (-1.20,0.00) .. (L);
    \draw[hed] (T) .. controls (-0.18,0.42) and (-0.48,-0.20) .. (L);

    \draw[hed] (T) .. controls (0.70,0.75) and (1.20,0.00) .. (R);
    \draw[hed] (T) .. controls (0.18,0.42) and (0.48,-0.20) .. (R);

    \draw[hed] (L) .. controls (-0.42,-0.22) and (0.42,-0.22) .. (R);
    \draw[hed] (L) .. controls (-0.45,-0.98) and (0.45,-0.98) .. (R);

    \node[vtx] at (T) {};
    \node[vtx] at (L) {};
    \node[vtx] at (R) {};
\end{scope}





\begin{scope}[xshift=0cm, yshift=-0.4cm]
    \filldraw[fill=gray!20, draw=black] (0,0) circle (1.90cm);

    %
    \coordinate (topL) at (-0.18,0.84);
    \coordinate (topR) at ( 0.18,0.84);

    \foreach \ang in {0,120,240}{
        \begin{scope}[rotate=\ang]

            \coordinate (A) at (-0.18,0.84);
            \coordinate (B) at (-0.817,-0.264);

            \draw[hed]
                (A)
                .. controls (-0.72,0.82) and (-1.12,0.16) ..
                (B);

            \draw[hed]
                (A)
                .. controls (-0.31,0.46) and (-0.48,0.02) ..
                (B);

            \node[vtx] at (A) {};
            \node[vtx] at (B) {};

        \end{scope}
    }

    \node[scale=1.2] at ( 90:1.52) {$Y$};
    \node[scale=1.2] at (210:1.52) {$Y$};
    \node[scale=1.2] at (330:1.52) {$Y$};
\end{scope}


\begin{scope}[xshift=4.4cm, yshift=-0.7cm]
    \filldraw[fill=gray!20, draw=black] (0,0) circle (1.90cm);

    \coordinate (To) at (0,0.92);
    \coordinate (Lo) at (-0.80,-0.46);
    \coordinate (Ro) at (0.80,-0.46);

    \coordinate (Ti) at (0,0.42);
    \coordinate (Li) at (-0.36,-0.21);
    \coordinate (Ri) at (0.36,-0.21);

    \draw[hed] (To) .. controls (-0.55,0.74) and (-0.92,0.10) .. (Lo);
    \draw[hed] (Lo) .. controls (-0.25,-0.84) and (0.25,-0.84) .. (Ro);
    \draw[hed] (Ro) .. controls (0.92,0.10) and (0.55,0.74) .. (To);

    \draw[hed] (Ti) .. controls (-0.22,0.30) and (-0.44,-0.01) .. (Li);
    \draw[hed] (Li) .. controls (-0.10,-0.36) and (0.10,-0.36) .. (Ri);
    \draw[hed] (Ri) .. controls (0.44,-0.01) and (0.22,0.30) .. (Ti);

    \node[vtx] at (To) {};
    \node[vtx] at (Lo) {};
    \node[vtx] at (Ro) {};
    \node[vtx] at (Ti) {};
    \node[vtx] at (Li) {};
    \node[vtx] at (Ri) {};

    \node[scale=1.2] at ( 90:1.49) {$X$};
    \node[scale=1.2] at (210:1.49) {$X$};
    \node[scale=1.2] at (330:1.49) {$X$};
\end{scope}



\begin{scope}[xshift=8.8cm, yshift=-0.4cm]
    \filldraw[fill=gray!20, draw=black] (0,0) circle (1.90cm);

    \begin{scope}[xscale=-1]

        \coordinate (Touter) at (0,0.92);
        \coordinate (Tinner) at (0,0.42);

        \coordinate (Ltop) at (-1.02,-0.37);
        \coordinate (Lbot) at (-0.82,-0.71);

        \coordinate (Rtop) at ( 1.02,-0.37);
        \coordinate (Rbot) at ( 0.82,-0.71);


        \draw[hed]
            (Touter)
            .. controls (-0.72,0.84) and (-1.17,0.23) ..
            (Ltop);

        \draw[hed]
            (Ltop)
            .. controls (-0.58,-0.02) and (-0.18,0.30) ..
            (Tinner);

        \draw[hed]
            (Tinner)
            .. controls (0.18,0.30) and (0.58,-0.02) ..
            (Rbot);

        \draw[hed]
            (Rbot)
            .. controls (0.26,-0.92) and (-0.26,-0.92) ..
            (Lbot);

        \draw[
            hed,
            preaction={draw=gray!20, line width=2.8pt}
        ]
            (Lbot)
            .. controls (0.02,-0.18) and (0.62,-0.06) ..
            (Rtop);

        \draw[hed]
            (Rtop)
            .. controls (1.17,0.23) and (0.72,0.84) ..
            (Touter);

        \node[vtx] at (Touter) {};
        \node[vtx] at (Tinner) {};
        \node[vtx] at (Ltop)   {};
        \node[vtx] at (Lbot)   {};
        \node[vtx] at (Rtop)   {};
        \node[vtx] at (Rbot)   {};

    \end{scope}

    \node[scale=1.2] at ( 90:1.55) {$X$};
    \node[scale=1.2] at (210:1.55) {$Z$};
    \node[scale=1.2] at (330:1.55) {$Y$};
\end{scope}
\end{tikzpicture}
\caption{A $4$-regular multigraph corresponding to the $\ket{K_3}$ state in the top, and three possible detachments at the bottom. The bottom three multigraphs have $3$, $2$ and $1$ connected components, matching the nullities $\nu(YYY)=2$, $\nu(XXX)=1$ and $\nu(XYZ)=0$, respectively.
}
\label{fig:detachment_example_appendix}
\end{figure}
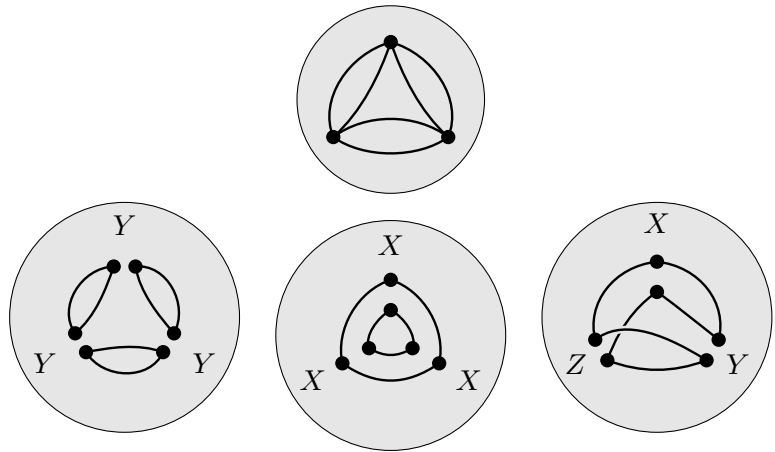

We note that this connection between circle graphs and $4$-regular graphs has been exploited before in the quantum information community~\cite{dahlberg2020counting, dahlberg2018transforming, dahlberg2022complexity, harrison2025fermionic}. In particular in \cite{dahlberg2020counting, dahlberg2018transforming, dahlberg2022complexity}, the authors studied Eulerian tours on $4$-regular multigraphs, which provide information on vertex-minors and the number of graphs locally equivalent to a given graph. An Eulerian tour of a graph $F$ is a closed walk in $F$ that traverses every edge of $F$ exactly once. Such an Eulerian tour on $F$ corresponds to a complete Pauli string $P$, such that $c(F||P) = 1$; Definition \ref{def:rank_null_4reg} then tells us that these correspond to complete Pauli strings with zero nullity, i.e.~single-qubit Pauli measurements which yield completely random outcomes. Furthermore, a complete Pauli string $P$ with $\nu(P)=0$ corresponds to a generating set of \emph{destabilizers} consisting only of single-qubit Pauli operators; see~\cite{aaronson2004improved}. Our work can be seen as an extension of the work from \cite{dahlberg2020counting, dahlberg2018transforming, dahlberg2022complexity}; in our work, we capture more information by not restricting to those complete Pauli strings with zero nullity.

\subsection{Entanglement measures of circle stabilizer states through $4$-regular multigraphs}
The previous subsection expressed the rank function of a $4$-regular multigraph $F$ in terms of the number of connected components of the detachment $F||P$. Unfortunately, the quantity $c\left(F||P\right)$ can be hard to understand. We show in this subsection that the rank function can be lower bounded, and sometimes even be explicitly expressed, in terms of \emph{circuit-partitions} of $F$. This allows us to obtain bounds on the two considered entanglement measures in terms of properties of \emph{circuit-partitions} of $F$. We start by recalling the definition of a circuit and a cycle in multigraphs.

\begin{definition}
A \emph{circuit} in a multigraph $F$ is a closed trail, i.e., a closed walk that does not repeat any edge. A \emph{cycle} is a circuit that does not repeat any vertex other than the initial vertex at the end. A \emph{circuit partition} of $F$ is a collection of pairwise edge-disjoint circuits whose union are the edges of $F$. A \emph{cycle partition} of $F$ is a circuit partition all of whose circuits are cycles.  
\end{definition}

We remind the reader that in a multigraph one can have circuits of length one (loops) and circuits of length two (parallel edges). Note that to specify a circuit partition, it suffices to choose a pairing of the $4$ half-edges at each vertex, i.e., a transition for each vertex. As such, circuit partitions on coded $4$-regular multigraphs are in bijection with complete Pauli strings, where the choice of $P_v$ at each vertex indicates the pairing chosen (see, e.g., Figure~\ref{fig:detachments}). This observation leads to the following definition.

\begin{definition}\label{def:number_circuits}
Let $F$ be a coded $4$-regular multigraph. The number of circuits in a circuit partition described by  a complete Pauli string $P$ is denoted by $\operatorname{cp}(F, P)$.
\end{definition}

We remark that for $P$ complete it holds that $c(F||P)=\textrm{cp}(F,P)$, where we define $\textrm{cp}(F, P)$ to be the number of connected components of the detachment graph $F||P$. That is,  $c(F||P)$ is exactly the number of circuits in the circuit partition of $F$ described by $P$. This is a consequence of the construction described in Definition~\ref{def:detachment}, where every circuit in $F$ becomes a connected component in $F||P$.

\subsubsection{Distance of circle stabilizer states in terms of $4$-regular multigraphs}

We now analyze the distance of circle stabilizer states in terms of $4$-regular multigraphs. Recall that the girth of a multigraph $F$ is the length of the shortest cycle in $F$ and is denoted by $g(F)$. The next result shows that the distance of a circle stabilizer state is upper-bounded by the girth of any $4$-regular multigraph associated with it.

\begin{lemma}\label{lemma:girth_distance}
Let $\ket{\psi}$ be a circle stabilizer state and let $F$ be a coded $4$-regular multigraph associated with it, i.e., a coded $4$-regular multigraph such that $r_{\ket{\psi}}=r_F$. Then it holds that $g(F)\geq d({\ket{\psi}})$.
\end{lemma}

\begin{proof}
Let $F$ be a coded $4$-regular multigraph associated with $\ket{\psi}$ such that $r_{\ket{\psi}}=r_F$ and $\nu_{\ket{\psi}}=\nu_F$. Let $\gamma$ be a shortest cycle of $F$, and let $U=V(\gamma)$ be its set of vertices. We construct a Pauli string $P$ with support $U$ as follows. For each vertex $v\in U$, the
cycle $\gamma$ uses exactly two of the four half-edges incident with $v$. Since $F$ is coded, for each vertex $v\in V$, there exists a bijection of each pairing of half edges with $\{X,Y,Z\}$. Choose $P_v\in\{X,Y,Z\}$ to be the transition such that one pairing consists of precisely
the two half-edges used by the cycles $\gamma$. For $v\notin U$, set $P_v=I$.

\medskip

Consider the detachment $F||P$. By construction, at every vertex $v\in U$ the two half-edges of $\gamma$ are placed on one of the two new vertices created from $v$, while the two remaining half-edges are placed on the other new vertex. Therefore, the edges of $\gamma$ form a connected component of $F||P$ which is separated from the rest of the graph. In
particular, the detachment increases the number of connected components. This implies that $\nu_{\ket{\psi}}(P)=\nu_F(P)=c(F||P)-c(F)\geq 1$. Thus, the local commutative subgroup $\mathcal I(\mathcal S_{\ket{\psi}},P)$ contains a non-identity stabilizer and we obtain that
\begin{align*}
    d(\ket{\psi})\leq w(P)=|U|\leq g(F).
\end{align*}
This concludes the proof of the lemma.
\end{proof}

The above result also follows from Theorems 5 and 6 in \cite{brijder2015isotropic}. We note that the distance can be strictly smaller than the girth (see Figure~7 in~\cite{brijder2015isotropic}).

\subsubsection{Locally accessible information of circle stabilizer states in terms of $4$-regular multigraphs}

Recall from Definition \ref{corr:alpha_char} that the locally accessible information is defined in terms of a maximization of the nullity function over complete Pauli strings. However, as previously discussed (see Definition \ref{def:number_circuits}), for complete Pauli strings, the correspondence between the rank function on $4$-regular multigraphs and circuits is particularly clean.

\begin{lemma}\label{lemma:nullity_complete_4reg}
Let $\ket{\psi}$ be any circle stabilizer state with an associated $4$-regular connected coded multigraph $F$ on vertices $V$. Let $P$ be a complete Pauli string. Then
\begin{align}
\nu_{\ket{\psi}}(P) = \operatorname{cp}(F, P)-1\ .
\end{align}
\end{lemma}

\begin{proof}
The equality follows from the fact that, for a Pauli string of full support, the number of components of $F||P$ is equal to the number of circuits $\textrm{cp}(F,P)$ associated with the transitions of $P$. Thus, by~\eqref{eq:nullity_F}, we have $\nu_{\ket{\psi}}(P)=\nu_{F}(P)=c(F||P)-c(F)=\operatorname{cp}(F,P)-1$.
\end{proof}

Note that the above equality can be seen as a circle stabilizer state version of the extended Cohn-Lempel equality (see Theorem 4 in \cite{traldi2011binary}), which relates the size of a circuit-partition $P$ of $F$ and the nullity of $I(G, P)$. As a consequence, the above lemma leads to the following characterization of the locally accessible information for circle stabilizer states.

\begin{corollary}\label{corr:alpha_circuit_partition}
Let $\ket{\psi}$ be any circle stabilizer state with an associated coded $4$-regular multigraph $F$ that is connected. Then the locally accessible information is the largest number of circuits in a circuit partition of $F$, minus one. That is
\begin{align*}
\alpha_{\rm loc}\left(\ket{\psi}\right) = \max_{\mathclap{\substack{P \textrm{ complete}}}}~\operatorname{cp}(F, P)-1\ .
\end{align*}
Moreover, the maximum circuit partition is a cycle partition, i.e., every circuit is a cycle.
\end{corollary}

\begin{proof}
By the characterization of $\alpha_{\rm loc}$ in terms of the nullity function,
we have
\begin{align*}
\alpha_{\rm loc}(\ket{\psi})
=
\max_{\mathclap{\substack{P \textrm{ complete}}}}
~\nu_{\ket{\psi}}(P).
\end{align*}
Since $F$ is a $4$-regular multigraph such that $\nu_{\ket{\psi}}=\nu_F$ and by Lemma~\ref{lemma:nullity_complete_4reg}, we have
\begin{align*}
    \alpha_{\rm loc}(\ket{\psi})=\max_{\mathclap{\substack{P \textrm{ complete}}}}\operatorname{cp}(F,P)-1.
\end{align*}
The moreover part follows from the fact that it is always possible to turn a non-cycle circuit into two smaller circuits by changing a transition.
\end{proof}

\section{Maximum entanglement in circle stabilizer states}\label{sec:max_ent_in_circ}

In this section we study extremal properties of the entanglement measure introduced in Definition~\ref{sec:ent_measures} for circle stabilizer states. An important result used throughout this section is the following generalization of the classical Moore bound found in~\cite{alon2002moore}.

\begin{theorem}[Average degree Moore bound]\label{thm:moore}
Let $G$ be a graph on $n$ vertices of average degree $r$ and girth $g$. Then it holds that $g\leq 2\log_{r-1} n+2$.
\end{theorem}

We now use Theorem~\ref{thm:moore} to provide a general upper bound for the distance of a circle stabilizer state.

\begin{theorem}\label{thm:distance_circle_graph}
The distance of a circle stabilizer state $\ket{\psi}$ on $n$ qubits satisfies $d(\ket{\psi}) = O\left(\log n\right)$.
\end{theorem}

\begin{proof}
Since $\ket{\psi}$ is a circle stabilizer state, by Proposition~\ref{prop:circle_characterization} there exists a coded $4$-regular multigraph $F$ on the same vertex set such that $r_{\ket{\psi}}=r_F$. By Lemma~\ref{lemma:girth_distance}, the distance of $\ket{\psi}$ is bounded above by the girth of $F$, that is, $d(\ket{\psi})\leq g(F)$. If $F$ contains parallel edges or self loops, i.e., it is a proper multigraph, then $g(F)\leq 2$ and we are done. Otherwise, by applying Theorem~\ref{thm:moore} with $r=4$ for the graph $F$, we obtain that $d(\ket{\psi})\leq g(F)=O(\log n)$.
\end{proof}

The above proof can be extended to any generalized Hamming weight $d_\ell$ for fixed $\ell$. We first note that, for a fixed $\ell$, the distance $d_\ell$ is well-defined for sufficiently large $n$. Indeed, every stabilizer state $\ket{\psi}$ is locally equivalent to a graph state $\ket{G}$. By Ramsey's theorem~\cite{graham1991ramsey}, if $G$ is sufficiently large, then $G$ contains either an independent set of
size $\ell$ or a clique of size $\ell+1$. In the first case, the graph $G$ itself has an independent set of size $\ell$. In the second case, after performing a local complementation at one vertex of the clique, the remaining $\ell$ vertices of the clique form an independent set. Hence, in both cases, some graph locally equivalent to $G$ contains an independent set of size $\ell$. By Proposition~\ref{prop:alpha_loc_independent}, this implies
$\alpha_{\rm loc}(\ket{\psi})\geq \ell$ and consequently $d_\ell(\ket{\psi})$ is well-defined. We remark that a similar argument was used in~\cite{ascoli2026almost} to define the so-called vertex-minor Ramsey number (see~\cite{ascoli2026almost, bae2026vertex}). We also note the following corollary of Theorem~\ref{thm:moore}. 

\begin{corollary}\label{lem:many_cycles}
Let $F$ be a $4$-regular multigraph on $n$ vertices. Then $F$ contains at least $n/(20\log n)$ edge-disjoint cycles, each of length at most $10\log n$ for sufficiently large $n$.
\end{corollary}

\begin{proof}
Set $L=10\log n$ to be the length of the cycles targeted throughout the proof and let $\mathcal{C}$ be a maximal collection of edge-disjoint cycles in $F$, each
of length at most $L$. We will show that $|\mathcal C|\geq n/2L=n/20\log n$. Suppose to the contrary that $|\mathcal C|<n/2L$. Let $\widetilde{F}$ be the multigraph obtained from $F$ by deleting the edges of all cycles in $\mathcal{C}$. Since $F$ is $4$-regular, it has $2n$ edges. Hence, the total number of edges in $\widetilde{F}$ is at least $2n-L\cdot |\mathcal{C}|>3n/2$. Thus the average degree of $\widetilde{F}$ is
at least $3$.

\medskip

We claim that $\widetilde{F}$ contains a cycle of length at most $L$. If $\widetilde{F}$ contains a loop or two parallel edges, then $\widetilde{F}$ contains a cycle of
length at most $2$, which is less than $L$ for all sufficiently large $n$. Otherwise, we may asume that $\widetilde{F}$ is simple and apply Proposition~\ref{thm:moore}. Since $\widetilde{F}$ has average degree larger
than $3$, we obtain
\begin{align*}
g(\widetilde{F})\leq 2\log n+2<10\log n=L.
\end{align*}
This finishes the proof of the claim. Since the new cycle is edge-disjoint from all cycles in $\mathcal{C}$, this contradicts the maximality of $\mathcal{C}$. Thus, $|\mathcal{C}|\geq n/2L$, which concludes the proof of the corollary.
\end{proof}

We now extend the argument from Theorem~\ref{thm:distance_circle_graph} to $d_\ell$ for fixed $\ell$.

\begin{theorem}\label{thm:gen_hamming_weight}
For any fixed $\ell \in \mathbb{N}_{>0}$, the generalized Hamming weight $d_\ell$ of a circle stabilizer state $\ket{\psi}$ satisfies $d_\ell(\ket{\psi}) = O\left(\log n\right)$.    
\end{theorem}

\begin{proof}
Since $\ket{\psi}$ is a circle stabilizer state, by Proposition~\ref{prop:circle_characterization} there exists a connected coded $4$-regular multigraph $F$ on the same vertex set such that $r_{\ket{\psi}}=r_F$ and $\nu_{\ket{\psi}}=\nu_F$. Moreover, by Corollary~\ref{lem:many_cycles}, there exists a collection $\mathcal{C}=\{C_1,\ldots,C_\ell\}$ of $\ell$
edge-disjoint cycles in $F$ such that $|C_i|\leq 10\log n$ for every $1\leq i\leq \ell$. Let $W$ be the union of the
vertices of these cycles. We now define a Pauli string $P$ supported on $W$. Note that
\begin{align*}
    |W|\leq\sum_{i=1}^\ell |C_i|\leq 10\ell \log n=O(\log n).
\end{align*}
For each vertex $v\in W$, choose the transition $P_v$ which pairs the two half-edges of each selected cycle passing through $v$. This is well-defined:
since the cycles are edge-disjoint and $F$ is $4$-regular, at most two selected cycles pass through a given vertex, and if two do, they use complementary pairs of half-edges. For the remaining vertices, i.e., $v\notin W$, set $P_v=I$.

\medskip

With this choice of $P$, each selected cycle $C_i$ becomes a connected component of the detachment $F||P$. Hence, because there are edges of $F$ in none of the cycles, the number of connected components satisfies $c(F||P)\geq \ell+1$. Since $F$ is connected, this gives by~\eqref{eq:nullity_F} that
\begin{align*}
\nu_{\ket{\psi}}(P)=\nu_F(P)=c(F||P)-1\geq \ell.
\end{align*}
In particular, this implies that $\dim(\mathcal{I}(\mathcal{S}_{\ket{\psi}},P))\geq \ell$. Thus, there exists a substring $Q$ of $P$ such that $\nu_{\ket{\psi}}(Q)$ and $w(Q)\leq w(P)=|W|=O(\log n)$. Therefore, by Definition~\ref{def:hamming_weights}, we have $d_\ell(\ket{\psi})\leq w(Q)=O(\log n)$. This concludes the proof of the theorem.
\end{proof}

Corollary~\ref{lem:many_cycles} can be also used to give a general lower bound to $\alpha_{\rm loc}$.

\begin{theorem}\label{thm:alpha_circle_graph}
A circle stabilizer state $\ket{\psi}$ satisfies $\alpha_{\rm loc}(\ket{\psi}) = \Omega\left(\frac{n}{\log n }\right)$.
\end{theorem}

\begin{proof}
Let $F$ be the associated coded $4$-regular multigraph with the circle stabilizer state $\ket{\psi}$ given by Proposition~\ref{prop:circle_characterization}. Hence, by Corollary~\ref{corr:alpha_circuit_partition}, the locally accessible information of $\ket{\psi}$ can be lower bounded by the size of the largest cycle partition of $F$. To bound such a partition, let $\mathcal{C}$ be the collection of edge-disjoint cycles obtained by Corollary~\ref{lem:many_cycles} and remove it from the graph $F$ to obtain a graph $H$. The graph $H$ still has all vertices of even degree, and thus by Veblen's theorem it has a decomposition into edge-disjoint cycles. We thus obtain a cycle partition of $F$ of size at least $|\mathcal{C}|=n/(20\log n)$. Hence, we have $\alpha_{\textrm{loc}}(\ket{\psi})\geq |\mathcal{C}|-1=\Omega(n/\log n)$.
\end{proof}

\section{Geelen's conjecture and entanglement in vertex-minor-closed families}

In this section we explain how Geelen's weak structural conjecture for vertex-minors can be used to  `lift' the bounds on $\alpha_{\rm loc}$ and the distance $d$ proved for circle stabilizer states to \emph{strongly rank-connected} arbitrary proper vertex-minor-closed families. The conjecture says, roughly, that the highly rank-connected members of any proper vertex-minor-closed graph class are circle graphs up to a bounded-rank perturbation. We first recall the graph-theoretic notions appearing in the
statement and introduce a more general---but equivalent---version of the conjecture in terms of stabilizer states. We then show that bounded-rank perturbations have only a controlled effect on the rank and nullity functions associated with single-qubit Pauli measurements. Combining this with our bounds for circle stabilizer states allows us to lift the result in Section~\ref{sec:max_ent_in_circ} to sufficiently rank-connected states in proper vertex-minor-closed families.

\subsection{Geelen's weak structural conjecture on vertex-minors}\label{sec:geelens_conj}

We will first define the notion of proper vertex-minor-closed families introduced in~\cite{mccarty2021local}. Recall that a graph $H$ is a vertex-minor of $G$ if $H$ is an induced subgraph of a graph $G'$ locally equivalent to $G$.

\begin{definition}[Proper vertex-minor-closed families of graphs]\label{def:PVMC_graphs}
A family of graphs $\mathcal{F}$ is a \emph{proper vertex-minor-closed} (PVMC) family if it is closed under taking vertex-minors and it is not the family of all graphs.
\end{definition}

Alternatively, a PVMC family $\mathcal{F}$ is a family obtained by forbidding a non-empty family as vertex-minors of $\mathcal{F}$. As discussed in Section~\ref{sec:prelim}, the concept of vertex-minors can be extended to stabilizer states (see Definition~\ref{def:vertex_minor_stab}). This naturally leads to the following definition.

\begin{definition}[Proper vertex-minor-closed families of stabilizer states]\label{def:PVMC}
A family of stabilizer states $\mathcal{F}$ is a \emph{proper vertex-minor-closed family} if it is closed under vertex-minors and it is not the family of all stabilizer states.
\end{definition}

From a quantum perspective, PVMC families correspond to classes of stabilizer states that are closed under local Clifford operations and single-qubit Pauli measurements, yet exclude certain states. Examples are given by stabilizer states with bounded \emph{rank-width} (which are states not useful for measurement-based quantum computation~\cite{van2007classical}), or those stabilizer states that can be created with a fixed number of photonic emitters~\footnote{This follows from the characterization from~\cite{li2022photonic}, where the smallest number of photonic emitters corresponds to the value of the smallest height function. This can be seen to match with the definition of the linear rank-width~\cite{oum2017rank}.}. PVMC families have also been studied from the context of measurement-based quantum computation (MBQC); the McCarty-Geelen conjecture states that MBQC on any PVMC family of states can be classically efficiently simulated~\cite{mccarty2021local,harrison2025fermionic}.

\medskip

Before stating Geelen's weak structure conjecture, we introduce the notion of rank perturbations and rank-connectivity~\cite{mccarty2021local, oum2023rank}.

\begin{definition}
    A \emph{rank-$p$ perturbation} of a graph $G$ is a graph whose adjacency matrix can be obtained by adding over $\mathbb{F}_2$ a symmetric matrix of rank at most $p$ to the adjacency matrix of $G$ and then changing all diagonal entries to~$0$.
\end{definition}

Two graphs have similar rank functions if they are related by a low-rank perturbation (which we make formal in lemma \ref{lemma:perturbations_affect_rank}). As such, their corresponding graph states have similar behavior under single-qubit Pauli measurements, and thus have similar entanglement. We note that it is also possible to extend the notion of rank-$p$ perturbations to stabilizer states directly (i.e.~by not referencing graphs), by using the notion of lifts and projections~\cite{geelen2008some}. We will not require this, however. Recall that the cut-rank function $\operatorname{cutrk}$ of a graph $G$ was introduced in Definition~\ref{def:cutrank_graph_states}.

\begin{definition}
    A graph $G$ on vertex set $V$ is \emph{$k$-rank-connected} if $\left|V\right|\geq 2k$ and $\operatorname{cutrk}(X) \geq \min(\left|X\right|, \left|V\setminus X\right|, k)$ for all $X\subseteq V$. The \emph{rank-connectivity} of a graph $G$ is the largest $k\in\mathbb{N}_0$ such that $G$ is $k$-rank-connected.
\end{definition}

In other words, rank-connectivity $k$ implies that any sufficiently balanced bipartition of $G$ has cut-rank at least $k$. We remark that $G$ is $1$-rank-connected if and only if $G$ is connected (and has more than one vertex). Note that the above definition generalizes to arbitrary stabilizer states through the extension of the cut-rank function in Corollary~\ref{corr:cutrk}. We are now ready to state Geelen's weak vertex-minor structure conjecture~\cite{mccarty2021local}.

\begin{conjecture}[Geelen's weak vertex-minor structure conjecture]\label{conj:geelen}
For every PVMC family $\mathcal{F}$ of graphs, there exist $k, p\in \mathbb{N}$ such that each graph in $\mathcal{F}$ that is $k$-rank-connected is a rank-$p$ perturbation of a circle graph.
\end{conjecture}

Note that the hypothesis of $\mathcal{F}$ being proper is necessary, since the family of all graphs never satisfies the statement of the conjecture. Since every concept used in the statement of Conjecture~\ref{conj:geelen} has an analogue for stabilizer states, we can rephrase the above conjecture as follows.

\begin{conjecture}[Geelen's weak vertex-minor structure conjecture for stabilizer states]\label{conj:geelen2}
For every PVMC family $\mathcal{F}$ of stabilizer states, there exist $k, p\in \mathbb{N}$ such that each state in $\mathcal{F}$ that is $k$-rank-connected is locally equivalent to a rank-$p$ perturbation of a circle graph state.
\end{conjecture}

Informally, the conjecture states that highly connected stabilizer states in any proper vertex-minor-closed family are essentially circle graph states, apart from a bounded-size rank perturbation.

\subsection{Entanglement under rank-$p$ perturbations}\label{sec:ent_under_rank_perturbations}

In this subsection, we show that rank-$p$ perturbations affect the entanglement in graph states in a controlled manner. More specifically, we first show that rank-$p$ perturbations can alter $\alpha_{\rm loc}$ by at most $p$. To do so we use the following lemma, which shows that rank-$p$ perturbations change the rank function by at most $p$ after a local relabellings of Pauli strings/single-qubit Clifford rotations.

\begin{lemma}\label{lemma:perturbations_affect_rank}
Let $G$ and $H$ be graphs that are rank-$p$ perturbations of one another, and let $r_{\ket{G}}, r_{\ket{H}}$ be their respective rank functions of their graph states. Then there exists a bijection $\phi:\mathcal{P}_n\rightarrow \mathcal{P}_n$ such that for all Pauli strings $P$, it holds that $\left|r_{\ket{G}}\left(P\right)-r_{\ket{H}}\left(\phi(P)\right)\right| \leq p$. Moreover, the bijection is a local relabelling of the Pauli operators, such that $\phi$ preserves the support of the Pauli strings, i.e., $\operatorname{supp}(P)=\operatorname{supp}(\phi(P))$.
\end{lemma}

\begin{proof}
Let $A_G$ and $A_H$ denote the adjacency matrices of $G$ and $H$, respectively. Since $H$ is a rank-$p$ perturbation of $G$, there exists a symmetric matrix $M$ over $\mathbb F_2$ such that $\operatorname{rank}(M)\leq p$ and $A_H$ is obtained from $A_G+M$ by changing all diagonal entries to zero. That is, 
\begin{align*}
    A_H=A_G+M+D,
\end{align*}where $D$ is the diagonal matrix whose $v$-th diagonal entry is $M_{vv}$.

\medskip

Set $W=\{v\in V:\:M_{vv}=1\}$. We define the bijection $\phi:\mathcal P_n\longrightarrow \mathcal P_n$ coordinate wise. At every vertex $v\notin W$, the Pauli operator is left
unchanged, while at every vertex $v\in W$ we exchange $X$ and $Y$ and leave $I$ and $Z$ unchanged. In other words,
\begin{align*}
\phi(P)_v=
\begin{cases}
P_v, & v\notin W,\\
Y, & v\in W \text{ and } P_v=X,\\
X, & v\in W \text{ and } P_v=Y,\\
P_v, & v\in W \text{ and } P_v\in\{I,Z\}.
\end{cases}
\end{align*}
The map $\phi$ is clearly a bijection. Moreover, it preserves the support and weight of every Pauli string.

\medskip

For a Pauli string $P \in \mathcal{P}_n$, let $I(G,P)$ and $I(H,\phi(P))$ be the matrices introduced in Definition~\ref{def:IAS}. Since $I(G,P)$ has exactly $w(P)$ columns, it follows by Lemma~\ref{lemma:substrings_kernel} that
\begin{align}\label{eq:rank_IGP}
r_{\ket G}(P)=w(P)-\nu_{\ket G}(P)=w(P)-\operatorname{nullity} I(G,P)=\operatorname{rank} I(G,P).
\end{align}
Hence, to compare the rank function of $\ket{G}$ and $\ket{H}$ we just need to compare the matrices $I(G,P)$ and $I(H,\phi(P))$.

\medskip

For any Pauli string $P$, let $N(P)$ be the $n\times \left|\textrm{supp}(P)\right|$ matrix with $v$'th column $N(P)_{v} = e_v$ if $P_v \in \lbrace{X, Y\rbrace}$ and zero else. The columns of the matrix $M\cdot N(P)$ are $Me_v$ if $P_v \in \lbrace{X, Y\rbrace}$ and zero else. A case analysis shows that $I(H,\phi(P))=I(G,P)+M\cdot N(P)$. Since the rank of $M\cdot N(P)$ cannot be larger than the rank of $M$ and $|\operatorname{rank}(X+Y)-\operatorname{rank}(X)|\leq \operatorname{rank}(Y)$ for any matrices $X$ and $Y$, we obtain by~\eqref{eq:rank_IGP} that
\begin{align*}
    \left|r_{\ket{G}}(P)-r_{\ket{H}}(\phi(P))\right|=\left|\operatorname{rank}I(G,P)-\operatorname{rank}I(H,\phi(P))\right|\leq \operatorname{rank}\left(M\cdot N(P)\right)\leq \operatorname{rank}{M}\leq p.
\end{align*}
This concludes the proof of the lemma.
\end{proof}

By running over all possible Pauli strings and by the fact that $\phi$ is a bijection, we immediately obtain the following corollary for $\alpha_{\rm loc}$.

\begin{corollary}\label{lemma:perturbations_alpha_kappa}
Let $G$ and $H$ be graphs that are rank-$p$ perturbations of one another. Then $\left|\alpha_\textrm{loc}(\ket{G})-\alpha_\textrm{loc}(\ket{H})\right| \leq p$. 
\end{corollary}

Another consequence of Lemma~\ref{lemma:perturbations_affect_rank} is the following result showing that rank-$p$ perturbations change the generalized Hamming weights $d_\ell$ in a controlled manner

\begin{lemma}\label{lemma:dist_gen_hamming_weights}
Let $\ket{G}$ be an $n$-qubit graph state with $\alpha\equiv\alpha_{\rm loc}(\ket{G})$, and generalized Hamming weights $d_1, d_2, \ldots, d_{\alpha}$. Let $H$ be a rank-$p$ perturbation of $G$ with $p<\alpha$. Then $d(\ket{H})\leq d_{p+1}$.
\end{lemma}
\begin{proof}
By definition of the generalized Hamming weights, there exists a Pauli string $P$ of weight $w(P)=d_{p+1}$ and nullity $\nu_{\ket{G}}\left(P\right)=p+1$. Let $\phi:\mathcal P_n\to\mathcal P_n$ be the support-preserving bijection given by Lemma~\ref{lemma:perturbations_affect_rank}. Since $\phi$ preserves support, it also preserves weight, and hence $w(\phi(P))=w(P)$. Moreover, Lemma~\ref{lemma:perturbations_affect_rank} also gives us that
\begin{align*}
\nu_{\ket{H}}(\phi(P))=w(\phi(P))-r_{\ket{H}}(\phi(P))\geq w(P)-(r_{\ket{G}}(P)+p)=\nu_{\ket{G}}(P)-p=1.
\end{align*}
Therefore, the locally commutative group $\mathcal I\bigl(\mathcal S_{\ket{H}},\phi(P)\bigr)$
contains a non-identity stabilizer $S$ with $\operatorname{supp}(S)\subseteq \operatorname{supp}(\phi(P))=\operatorname{supp}(P)$. 
Consequently, by the definition of distance, it follows that
\begin{align*}
    d(\ket{H})\leq w(S)\leq w(P)=d_{p+1}
\end{align*}
as desired.
\end{proof}

\subsection{Lifting to PVMC families}

Using the previous two lemmas, we are ready to lift the constraints on the asymptotic entanglement in circle stabilizer states (see theorems~\ref{thm:distance_circle_graph} and \ref{thm:alpha_circle_graph}) to arbitrary sufficiently connected stabilizer states in PVMC families.

\begin{theorem}\label{theorem:vm_closed_has_small_d}
Geelen's weak vertex-minor structure conjecture implies the following. Let $\mathcal{F}$ be a proper vertex-minor-closed class of stabilizer states and $k$ be the associated parameter in Geelen's conjecture attached to $\mathcal{F}$. Let $\ket{\psi_n}$ be an $n$-qubit stabilizer state in $\mathcal{F}$. Then either $\ket{\psi_n}$ is not $k$-rank-connected, or the following two statements hold
\begin{enumerate}
        \item[$(i)$] the locally accessible information $\alpha_{\rm loc}$ satisfies $\alpha_{\rm loc}\left(\ket{\psi_n}\right)=\Omega\left(\frac{n}{\log n }\right)$, and
             \item[$(ii)$] the distance satisfies $d(\ket{\psi_n})=O\left(\log n \right)$.
\end{enumerate}
\end{theorem}
\begin{proof}
Assume Geelen's weak vertex-minor structure conjecture, and let $k,p\in\mathbb N$ be the parameters associated with the proper vertex-minor-closed family $\mathcal F$ in Conjecture~\ref{conj:geelen2}. Suppose that $\ket{\psi_n}$ is $k$-rank-connected. We claim that both conclusions hold. By Conjecture~\ref{conj:geelen2}, there exist a circle graph $G_n$ and a graph $H_n$ such that $H_n$ is a rank-$p$ perturbation of $G_n$ and $\ket{\psi_n}$ is locally equivalent to the graph state $\ket{H_n}$. Since local Clifford unitaries act coordinate wise on Pauli strings and preserve their supports, both the distance and the locally accessible information are invariant under local equivalence. Therefore, it holds that  $d(\ket{\psi_n})=d(\ket{H_n})$ and $\alpha_{\rm loc}(\ket{\psi_n})=\alpha_{\rm loc}(\ket{H_n})$.

\medskip

We first prove~$(i)$. By
Corollary~\ref{lemma:perturbations_alpha_kappa}, it holds that $\left|\alpha_{\rm loc}(\ket{G_n})-\alpha_{\rm loc}(\ket{H_n})\right|\leq p$. Since $\ket{G_n}$ is a circle stabilizer state, Theorem~\ref{thm:alpha_circle_graph} gives $\alpha_{\rm loc}(\ket{G_n})=\Omega\left(\frac{n}{\log n}\right)$. As $p$ is a constant depending only on $\mathcal F$, it follows that
\begin{align*}
\alpha_{\rm loc}(\ket{\psi_n})=\alpha_{\rm loc}(\ket{H_n})\geq\alpha_{\rm loc}(\ket{G_n})-p=\Omega\left(\frac{n}{\log n}\right).
\end{align*}

We now prove~$(ii)$. By the preceding estimate,
$\alpha_{\rm loc}(\ket{G_n})>p$ for all sufficiently large $n$. Hence the generalized Hamming weight $d_{p+1}(\ket{G_n})$ is well-defined. Since $H_n$ is a
rank-$p$ perturbation of $G_n$, Theorem~\ref{thm:gen_hamming_weight} combined with Lemma~\ref{lemma:dist_gen_hamming_weights} yields that
\begin{align*}
d(\ket{\psi_n})=d(\ket{H_n})
\leq
d_{p+1}(\ket{G_n})=O(\log n).
\end{align*}
This proves the dichotomy stated in the theorem.
\end{proof}


\end{document}